\documentclass[
reprint,
superscriptaddress,
tightenlines,
aps,
pre,
]{revtex4-2}

\usepackage[colorlinks=true,
						linkcolor=blue,
						urlcolor=blue,
						citecolor=blue,
						bookmarks=true,
						pdfborder={0 0 0}]{hyperref}
\usepackage{lipsum}
\usepackage{amsmath,amssymb,amsthm,tikz}
\usepackage{nicefrac}
\usepackage{slashed}
\usepackage{amsfonts}
\usepackage{empheq}
\usepackage{dsfont}
\usepackage{xcolor}
\usepackage[shortlabels]{enumitem}

\newtheorem{theorem}{\sc Theorem}[section]

\newtheorem{proposition}[theorem]{\sc Proposition}

\numberwithin{equation}{section}
\theoremstyle{remark}

\newcommand{\be}{\begin{equation}}
\newcommand{\ee}{\end{equation}}

\def\bE{\mathbb{E}}

\def\bP{\mathbb{P}}

\def\t{\triangle}
     
\def\E{\bE}
\def\P{\bP}

\definecolor{darkgreen}{rgb}{0.0,0.5,0.0}
\definecolor{darkblue}{rgb}{0.0,0.0,0.3}
\definecolor{nicosred}{rgb}{0.65,0.1,0.1}
\definecolor{light-gray}{gray}{0.7}

\newcommand{\ER}{Erd\H{o}s-R\'enyi }

\begin{document}
\title{Insights from exact social contagion dynamics\\ on networks with higher-order structures}

\author{Istv\'an Z. Kiss}
\email{istvan.kiss@nulondon.ac.uk}
\affiliation{Network Science Institute, Northeastern University London, London E1W 1LP, United Kingdom}
\author{Iacopo Iacopini}
\affiliation{Network Science Institute, Northeastern University London, London E1W 1LP, United Kingdom}
\author{P\'eter L. Simon}
\affiliation{Numerical Analysis and Large Networks Research Group, Hungarian Academy of Sciences, Hungary}
\author{Nicos Georgiou}
\affiliation{Department of Mathematics, University of Sussex, Falmer, Brighton BN1 9QH, United Kingdom}

\date{\today} 
\begin{abstract} 
Recently there has been an increasing interest in studying dynamical processes on networks exhibiting higher-order structures, such as simplicial complexes, where the dynamics acts above and beyond dyadic interactions. Using simulations or heuristically derived epidemic spreading models it was shown that new phenomena can emerge, such as bi-stability/multistability. Here, we show that such new emerging phenomena do not require complex contact patterns, such as community structures, but naturally result from the higher-order contagion mechanisms. We show this by deriving an exact higher-order SIS model and its limiting mean-field equivalent for fully connected simplicial complexes. Going beyond previous results, we also give the global bifurcation picture for networks with 3- and 4-body interactions, with the latter allowing for two non-trivial stable endemic steady states. Differently from previous approaches, we are able to study systems featuring interactions of arbitrary order. In addition, we characterise the contributions from higher-order infections to the endemic equilibrium as perturbations of the pairwise baseline, finding that these diminish as the pairwise rate of infection increases. Our approach represents a first step towards a principled understanding of higher-order contagion processes beyond triads and opens up further directions for analytical investigations.
\end{abstract}

\maketitle	

\onecolumngrid

\section{Introduction}
Complex networks provide a powerful representation for the backbone of complex systems by describing their structural connections in terms of nodes, i.e., individuals, that interact through links~\cite{albert2002statistical, newman2003structure, barrat2008dynamical, latora_nicosia_russo_2017}. Recently there has been an increasing interest in non-pairwise approaches to networked populations~\cite{lambiotte2019networks, battiston2021physics, bick2023higher}. In fact, many real-world systems are composed by elements that interact in groups of different sizes, stretching the order of these fundamental interactions beyond the dyads. This is particularly true for social systems, where most social interactions involve more than two people at the time ---while links, by definition, can connect only pairs of individuals~\cite{wasserman1994social}. Higher-order network approaches can instead be used to explicitly encode the many-body social interactions that mediate the daily communication of human and non-human animal societies~\cite{patania2017shape, benson2018simplicial, iacopini2023not}. Hypergraphs and simplicial complexes, differently from networks, allow for interactions between any number of units, and can therefore provide a more accurate ``higher-order'' description for those systems featuring group interactions~\cite{battiston2020networks, torres2021and}. 

Modelling efforts in this directions have already shown that landmark dynamical processes on networks can behave very differently when we let nodes interact in groups composed by more than two individuals at the time~\cite{skardal2019abrupt, millan2020explosive, neuhauser2020multibody, alvarez2021evolutionary, schawe2022higher}.
Among these, extensions of adoption processes beyond pairwise interactions led to the appearance of new phenomenologies, such as critical mass effects and bi-stability in social contagion and norm emergence, whose presence depends on both dynamical and structural quantities~\cite{barrat2022social}.
For example, groups can amplify small initial opinion biases accelerating consensus formation in voter models~\cite{papanikolaou2022consensus}. In higher-order spreading processes it has been shown that the inclusion of peer pressure coming from group interactions can change the nature of the phase transition from an epidemic-free to an endemic state, allowing for their co-existence~\cite{iacopini2019simplicial, matamalas2020abrupt, de2020social}. Structural features, as heterogeneity, can suppress the onset of bi-stability ~\cite{landry2020effect, st2022influential}, or even lead to multi-stability and intermittency when the system presents community structure and groups follow a critical-mass dynamics~\cite{ferraz2023multistability}. Group size plays also an important role. A disease spreading through higher-order mechanisms tends to be concentrated and sustained by the largest groups~\cite{st2021master}. By contrast, when groups modulate the diffusion of social conventions starting from a seed of committed agents, a non-monotonic dependence on the group size has been found~\cite{iacopini2022group}. 

Despite the interesting insights obtained via the inclusion of these higher-order mechanisms, rigorous theoretical studies in this directions have been very limited, and most analytical treatments either rely on heuristic methods for their derivation or on approximations that are true only under certain conditions; this is the case for homogeneous and heterogeneous mean-field models~\cite{iacopini2019simplicial, jhun2019simplicial, landry2020effect}, pair-based and triadic approximations~\cite{malizia2023pair, burgio2023triadic}, microscopic Markov-chain approaches~\cite{matamalas2020abrupt, burgio2021network}, and approximate master equations~\cite{st2021master, st2022influential, burgio2023adaptive}.

Here, we devise a principled analytical formulation of the simplicial contagion model~\cite{iacopini2019simplicial} on fully connected higher-order structures and investigate its emerging dynamics. While most previous studies focused on the dynamics mediated by 2-body (1-simplices) and 3-body (2-simplices), we study the model in details up to 4-body interactions, i.e., when the infection spreads in complete 3-complexes. We then show, starting from the resulting Kolmogorov forward equations, how one can derive mean-field models for complete simplicial 2- and 3-complexes, and we prove their exactness for a number of nodes that tends to infinity. Going even further, we conjecturing a mean-field equation for simplicial contagion on a complete structure up to an arbitrary order $M$, a complete simplicial $M$-complex.
In addition, we show how one can map the additional infection pressure brought by higher-order interactions as perturbations of the classic system where infections are simply mediated by traditional pairwise links (1-simplices). 
Finally, through an extensive bifurcation analysis of the contagion dynamics on a system up to 4-body interactions, we show that, in addition to the bi-stability already found when adding 3-body contributions to a spreading dynamics, the phenomenology at higher-orders is further enriched and multi-stability can emerge ---where more than one non-trivial stationary state can co-exist. 

The paper is structured as follows.
In section~\ref{sec:model}, we formulate the exact Susceptible-Infected-Susceptible (SIS) contagion dynamics on fully connected structures. We formulate the Kolmogorov forward equations for a simplicial contagion on a complete simplicial 2- and 3-complex, derive the associate mean-field models, and prove that they are exact in the termodinamic limit.
We close the section with a conjecture on the most general case of a complete simplicial $M$-complex.
In section~\ref{sec:bif}, we present a full bifurcation analysis of the resulting mean-field models for complete simplicial 2- and 3-complexes. For the case up to 2-complexes, we study the full phase diagram and the stability of the solutions as a function of the infection parameters. The system presents two distinct bifurcation scenarios, either a traditional transcritical bifurcation, or a fold bifurcation leading to bi-stability. We then repeat the analysis for contagion dynamics running on complete simplicial 3-complexes, finding a much richer scenario where the system can display the transcritical behaviour, bi-stability, and fold bifurcation with the fold placed both before or after the transcritical point (multistability).
Finally, in section~\ref{sec:disc} we conclude by interpreting our findings and give a brief overview of the possible next challenges.

\section{From exact stochastic to limiting mean-field models}
\label{sec:model}

In this section we lay down the basic formulation for an exact simplicial SIS-like dynamics on a fully connected simplicial complex composed by $N$ interconnected nodes.

We represent the individuals of a social system as a simplicial complex, that is a collection of $k$-simplices~\cite{salnikov2018simplicial, battiston2020networks, torres2021and}. Each $k$-simplex (where $k$ denotes the order) represents a group interaction among $k+1$ nodes (size $k+1$). Under this framework, nodes are called 0-simplices, pairwise links are 1-simplices, etc. In addition, in order for a simplicial complex to be valid it requires downward closure, that is that all the sub-sets of it simplices need to be also part of the complex~\cite{aleksandrov1998combinatorial, hatcher2002algebraic}. For example, if the complex contains the 2-simplex $[i, j, k]$, it must also contain the lower-order combinations $[i, j]$, $[j, k]$, $[i, k]$, $[i]$, $[j]$ and $[k]$. In the social context, it means that whenever a group of individuals is having, for example, a conversation, all the possible sub-groups contained are also assumed to be interacting. 

We then consider a contagion dynamics that follows the simplicial contagion model introduced in Ref.~\cite{iacopini2019simplicial}, according to which an infection can spread in a population at different rates through group interactions of different sizes. More precisely, in addition to the traditional infection along links, or pairwise contagion, we allow for contagion events through simplices (groups of nodes) of arbitrary order ---as allowed by the size of the complex. Nodes can belong, as per a traditional SIS model, to two compartments: susceptible nodes ($S$) that can acquire an infection upon a contact with an infectious node; infectious nodes ($I$) that can pass the infectious to susceptible ones in the neighbourhood. While spontaneous recovery transitions from $I$ to $S$ are not affected by the higher-order structure, contagion events controlling the transitions from $S$ to $I$ can happen at different orders. In particular, a susceptible node receives a stimulus from every ``active'', or ``contagious'' simplex it is part of. A simplex is considered to be infectious if all the nodes composing it, except a single susceptible one, are infectious as well. This is clear in the examples shown in Fig.~\ref{fig:diagram}, where the higher-order contagion dynamics is explained for the case of a $S$ node $i$ belonging to a 3-simplex under three different scenarios of (rows). Given the nested structure imposed by the simoplicial complex, the maximal infection pressure is reached when all the other nodes in the simplex are infectious. 

We now derive the exact forward Kolmogorov equations for the case of a simplicial contagion running exclusively via pairwise (1-simplices) and three-body interactions (2-simplices), and we show how a mean-field limit can be derived starting from them. Thereafter, we rigorously show that this model is exact in the limit of $N\rightarrow \infty$.
Using the same approach, we then extend the derived SIS model up to four-body interactions (3-simplices), deriving again the associated mean-field formalism (the proof of its exactness is not given as it naturally follows the same idea as for the previous case up to 2-simplices). Finally, we conjecture the mean-field limit for the most general formulation of the model, that is when infection is mediated by higher-order interactions over arbitrary $k$-simplicies, $k$ being bounded only be network size, i.e. $1 \le k \le N-1$. Table~\ref{tab:table_parameters} provides a recap of the notation that will be used from now on to denote infection rates associated to simplices of different order.

\begin{figure}[t]
     \centering
     \includegraphics[width=\linewidth, scale=0.33]{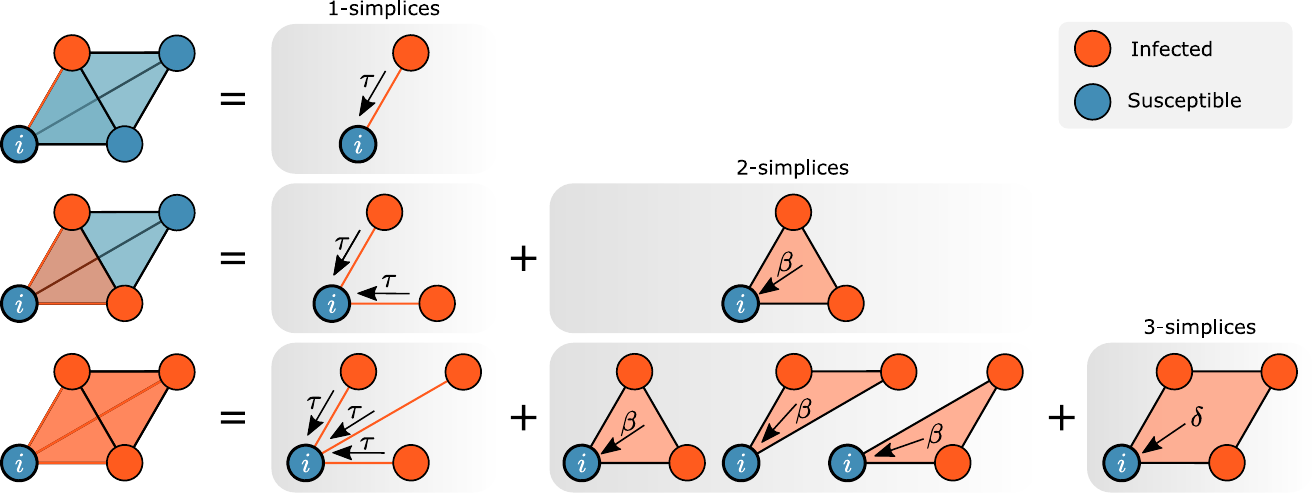}
     \caption{Schematic illustration of the different possible channels of infection of the simplicial contagion model. The susceptible (S) node $i$ participates to the depicted 3-simplex (left column); being part of a simplicial complex, by definition all the sub-faces of the simplex are also present. The infection pressure that $i$ receives grows with the number of infectious nodes (I) within the considered simplex.
     Top row: one node is infectious, thus $i$ can be infected by a single 1-simplex (link) with $\tau>0$.
     Middle row: two nodes are infectious, thus $i$ can be infected either by the two 2-simplices or by the ``infectious'' 2-simplex (triangle) with $\beta>0$.
     Bottom row: all three nodes are infectious, thus $i$ can be infected either by the three 1-simplices, by the three 2-simplices that they compose, or by the 3-simplex (tetrahedron) with $\delta>0$.}
     \label{fig:diagram}
\end{figure}

\begin{table}[b]
\center
\begin{tabular}{|c|c|c|c|}
          \hline
          & group size & infectivity & rescaled infectivity  \\
          \hline
1-simplex & 2    & $r_1=\tau$      & $s_1=\lambda$             \\
2-simplex & 3    & $r_2=\beta$     & $s_2=\mu$                 \\
3-simplex & 4    & $r_3=\delta$    & $s_3=\theta$              \\
\vdots &   \vdots  &   \vdots  &     \vdots          \\
$k$-simplex & $k+1$    &  $r_k$   & $s_k$              \\
          \hline
\end{tabular}
\caption{
Parameters associated to the infection coming from simplices of different order.}
\label{tab:table_parameters}
\end{table}

\subsection{Higher-order $SIS$ epidemics on a complete simplicial 2-complex}

Let us start from the easiest case of simplicial contagion where there are only two possible mechanisms of transmission: via pairs (1-simplices) and via ``triangles'' (2-complexes). Given the simplicial SIS model defined above, let us now focus on the possible transitions between the two classes of individuals.
If a node is infected (i.e.\ it is at state $I$), then it switches to state $S$ at a recovery rate $\gamma > 0$. When the number of infected nodes in the simplicial complex is $k$, one of them will switch to $S$ with a rate 
\be
c_k = \gamma k.
\ee

If a node is in state $S$, then it switches to state $I$ at a rate proportional to the number of infected nodes $I$ it is linked to---via simplices of different order; in particular, if the $S$ node is linked to $k$ infected ones via 1-simplices, then it switches to $I$ at a rate $\tau k $, $\tau > 0$. In this case, the number of infected nodes goes up by 1 with a rate of 
\be\label{eq:ak}
a_k = \tau k (N -k). 
\ee

In addition, three-body interactions provide extra infection pressure to susceptible individual $S$ (peer pressure). Since the simplicial complex is complete, any distinct pair of infected neighbours of the selected node would act as a an infectious 2-simplex that could also infect the susceptible node. Assuming that there are $k$ infected neighbours, there is an infection event for that node at an extra rate 
\be
b^\t_k = \beta {k \choose 2}, \quad \beta>0.
\ee
Since we are on the fully connected structure, the $k$ infected nodes are neighbours to all the susceptible ones. As such, the number of $I$ goes up due to the effect of the 2-simplices with a rate 
\be\label{eq:bk}
\beta^\t_k =  \beta {k \choose 2}(N - k), \quad \beta > 0.
\ee

After calling $I_t^{N}$ the number of infected nodes at time $t$, let us define
\be
p_k(t) = \P\{ I^N_t =k \}, 0 \le k \le N.
\ee

From the symmetry of the structure, due to its full connectivity, the epidemics (i.e.\ the number of infectious nodes at any given $t$) can be mapped to a branching process whose Kolmogorov equations are given by 

\be \label{eq:kolm-kn} 
\frac{d}{dt} p_k(t) =\begin{cases} 
(a_{k-1}+ \beta^\t_{k-1}) p_{k-1}(t) - (a_k + c_k + \beta^\t_k) p_{k}(t) +   c_{k+1}p_{k+1}(t), &\quad  3 \le k \le N-1,\\
a_{k-1} p_{k-1}(t) - (a_k + c_k + \beta^\t_k) p_{k}(t) +   c_{k+1}p_{k+1}(t), &\quad  k = 2,\\
 - (a_k + c_k) p_{k}(t) +   c_{k+1}p_{k+1}(t), &\quad  k = 1,\\
-c_k p_{k}(t) +   c_{k+1}p_{k+1}(t), &\quad  k = 0,\\
(a_{k-1}+ \beta^\t_{k-1}) p_{k-1}(t) - (a_k + c_k) p_{k}(t), &\quad  k = N.
\end{cases}
\ee 

Note that the Eq.~\eqref{eq:kolm-kn} can be reduced to 
\be \label{eq:kolm-kn2}
\frac{d}{dt} p_k(t) = (a_{k-1}+ \beta^\t_{k-1}) p_{k-1}(t) - (a_k + c_k + \beta^\t_k) p_{k}(t) +   c_{k+1}p_{k+1}(t), \quad 0 \le k \le N
\ee
subject to the boundary conditions 
\be
a_0 = a_{-1} = 0, \quad \beta^\t_0 = \beta^\t_1 =  \beta^\t_{-1}=  0, \quad c_{N+1} = 0.
\ee

The expected number of infected nodes $m^N(t) = \E(I^N_t) = \sum_{k=0}^{N}kp_k(t)$ at time $t$, after an index shift shown in Appendix~\ref{App:2-simp-proof}, satisfies the equation
\be \label{eq:prelimit}
\frac{d}{dt} m^N(t) =  \sum_{k=0}^N k [a_{k-1}p_{k-1}(t) - (a_k + c_k)p_k(t) + c_{k+1}p_{k+1}(t)] + \beta \sum_{\ell = 2}^{N-1} (N -\ell) {\ell \choose 2} p_{\ell}(t).
\ee

We can always find a stochastic process $\eta^N_t$ so that 
\be
\label{eq:rep1}
\frac{I^N_t}{N} = \frac{1}{N} m^N(t) + \eta^N_t.
\ee 
Equation \eqref{eq:rep1} can be rearranged as 
\be
\frac{I^N_t-m^N(t)}{N}= \eta^N_t,
\ee
which defines $\eta^N_t$ as the error between the density of infected individuals and their mean. If one assumes the existence of a hydrodynamic limit for $I^N_t/N$ so that it is concentrated around its mean, it would imply that we can take the limit as $N \to \infty$ in Eq.~\eqref{eq:rep1}, giving that, in the context of weak convergence, the randomness vanishes  in the limit, i.e. 
\be
\lim_{ N \to \infty} \frac{I^N_t}{N} = \lim_{ N \to \infty} \frac{1}{N} m^N(t)=m_{I}(t), \qquad  \lim_{N \to \infty} \eta^N_t = 0.
\ee
For $\eta^N_t$ to converge unscaled to $0$, it must be that Var$(\eta^N_t) \to 0$ as $N \to \infty$, so that it can converge to a constant. Since $\eta^N_t$ is uniformly bounded, the weak convergence to 0 implies the convergence of expectations and can be upgraded to an $\mathcal L^3$ convergence, so
\be\label{eq:ass2}
 \lim_{N \to \infty} \E( \eta^N_t) = \lim_{N \to \infty} \E((\eta^N_t)^p) = 0.
\ee 

Moreover, the limiting $m_I(t)$ will come out as the scaling limit of \eqref{eq:prelimit} and it is a density driven process. The coefficients in that equation also need to scale with $N$, otherwise the large $N$-limit will be trivially 0 or $\infty$. The correct scalings for $\tau$ and $\beta$ are already hidden in Equations~\eqref{eq:ak} and \eqref{eq:bk}; they need to scale proportionally to the number of simplices they correspond to. The infectivity parameter $\tau$ acts on links, and the simplicial complex contains $O(N^2)$ of them; similarly, $\beta$ acts on 2-simplices, and these are $O(N^3)$. To this end, we introduce two positive parameters $\lambda$ and $\mu$ so that 
\be\label{lambda-mu}
\tau = \lambda N^{-1}, \qquad \beta = \mu N^{-2}.
\ee
The extra $N$ power needed to balance $\tau$ and $\beta$ with the number of corresponding simplices is coming from the division with $N$ in Equation~\eqref{eq:rep1}. 

\begin{theorem}[Hydrodynamic limit for 2-simplices infections in fully connected simplicial 2-complexes] 
\label{thm:hydroKn}
Fix a time horizon $T>0$ and assume that the following uniform pointwise limit 
\be \label{eq:unipo}
 \lim_{ N \to \infty} \sup_{t \le T} \Big\|\frac{1}{N} m^N(t) - m_I(t)\Big\|_{\infty} = 0
\ee
exists as a deterministic function on $[0,1]$. Furthermore, assume the $\eta^N_t$, defined by \eqref{eq:rep1}, satisfies 
\be \label{eq:ass0}
\lim_{N \to \infty} \sup_{t \le T} \E(\eta^N_t) = 0.
\ee
Finally, let $\gamma$ denote the recovery rate and assume \eqref{lambda-mu}.

Then we have the weak convergence 
\be \label{eq:lim0}
\lim_{ N \to \infty} \frac{1}{N} I^N_t = m_I(t),
\ee
and $m_I(t)$ solves the ODE 
\be \label{eq:ODE1}
\frac{d}{dt}m_I(t) = (\lambda - \gamma) m_I(t) + \Big( \frac{\mu}{2} - \lambda\Big) m_I(t)^2 -\frac{\mu}{2} m_I(t)^3. 
\ee
\end{theorem}

The detailed proof of Theorem~\ref{thm:hydroKn} is given in Appendix~\ref{App:2-simp-proof} where the mathematical reason for the $\beta, \tau$ scalings becomes apparent.
In Figure~\ref{fig:two_close_curves} (left) we numerically show that the agreement between the Kolmogorov equations and their mean-field limit is excellent. Furthermore, it is evident that the added peer pressure brought by the 2-simplices elevates the endemic equilibrium, but the same increase in the value of the transmission rate across 2-simplices as for the pairwise rate of transmission leads to a less marked effect on the values of the endemic equilibrium. This suggests that pairwise transmission remains the main driver of the epidemics, especially for high values of the 1-simplices infectivity.

\begin{figure}[t]
     \centering
     \includegraphics[width=0.49\linewidth]{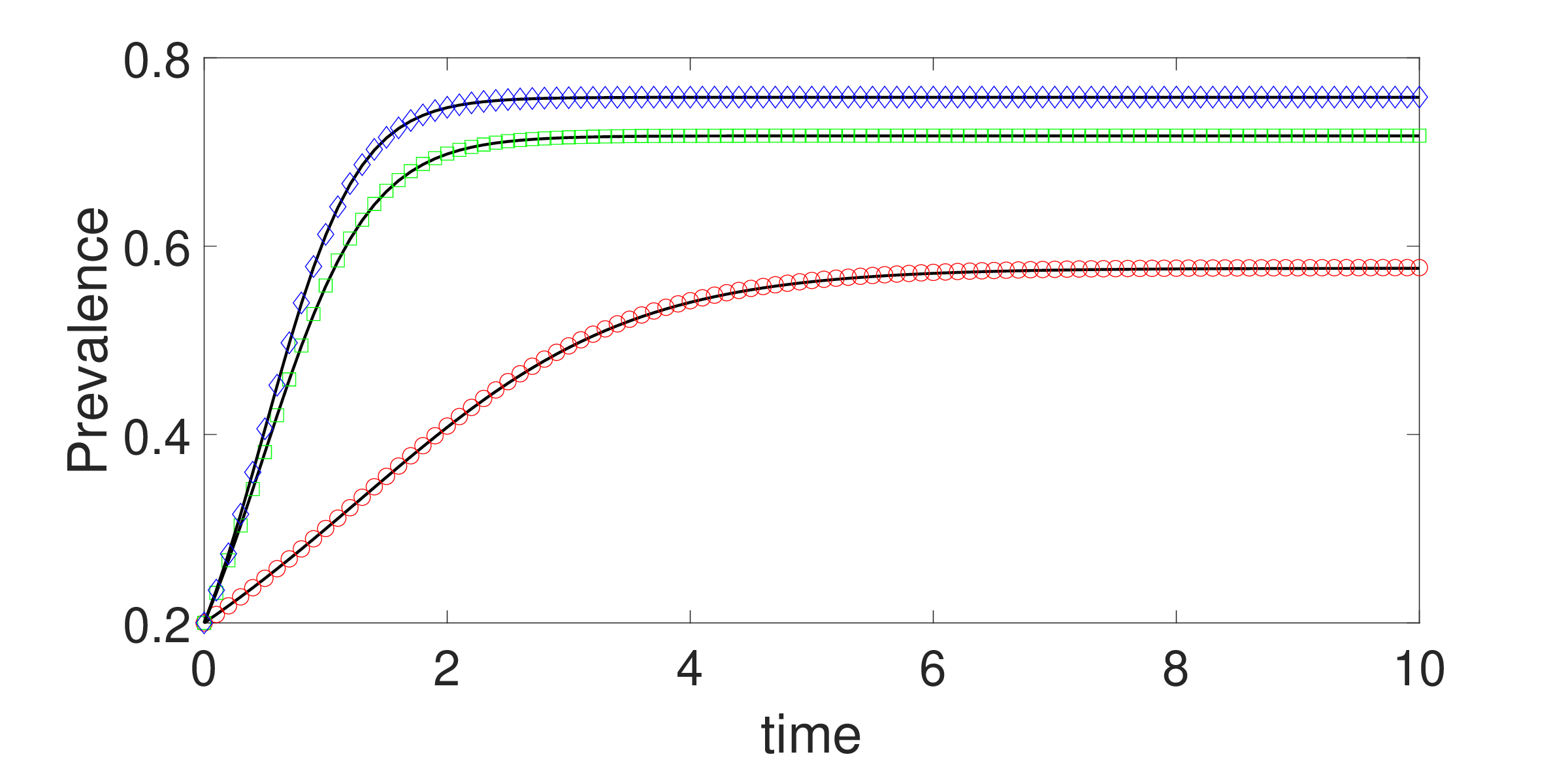}
     \includegraphics[width=0.49\linewidth]{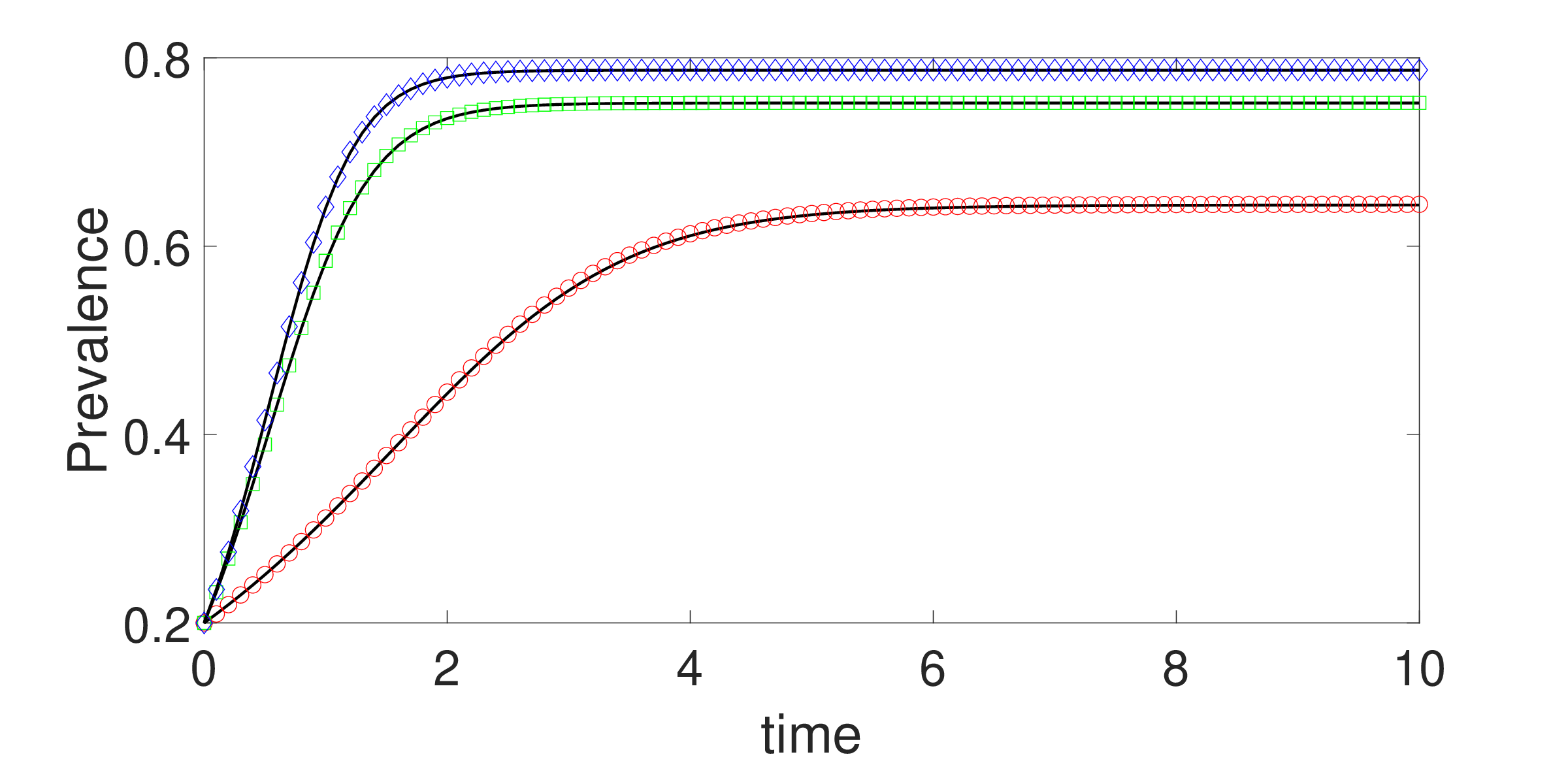}
     \caption{Comparison between the expected fraction of infected nodes based on solving the forward Kolmogorov equations (continuous lines) and the solution of the mean-field equation (markers). The case of 1- and 2-simplices (Kolmogorov equations~\eqref{eq:kolm-kn2} and mean-field equation~\eqref{eq:ODE1})  and 1-, 2- and 3-simplices (Kolmogorov equations~\eqref{eq:kolm-kn3} and mean-field equation~\eqref{ode3simplex}) are shown in the left and right panel, respectively. The following parameter combinations: leftmost panel ($\lambda, \mu$)=($1.5, 3$), ($3, 1.5$), ($3,3$) corresponding to increasing level of prevalence. The same values apply for the rightmost panel, with a fixed $\theta=5$ for all three cases. For this numerical tests we considered $N=2500$ nodes, and the epidemic is started with a seed 500 infected nodes.}
     \label{fig:two_close_curves}
\end{figure}

\subsection{Higher-order $SIS$ epidemics on a complete simplicial 3-complex}

We now follow the steps of the proof of Theorem \ref{thm:hydroKn}, but extending the previous case study up of one order. This means that the simplicial complex includes also 4-body interactions (3-simplices) that also act as possible channels of infection for susceptible nodes. The proof goes along the same lines, but it also highlights how the scaling of the rates needs to inversely scale with the number of simplices composing the complex. 
We want to keep the infection pressure from 1- and 2-simplices from before (though it can be set to 0 if necessary) and also add an extra pressure $\delta$ due to transmission mediated by 3-simplices $I-I-I-S$.

As before, if a node is infected (i.e.\ it is at state $I$), then it switches to state $S$ at a recovery rate $\gamma > 0$. 
When the number of infected is $k$, one of them will switch to $S$ with a rate 
\be
c_k = \gamma k.
\ee
If a node is at state $S$ then it switches to state $I$ at rate proportional to the number of infected nodes $I$ it is linked to; in particular if $S$ is linked to $k$ infected, then it switches to $I$ at a rate $\tau k $, $\tau > 0$. Then the number of infected goes up by 1 with  a rate of 
\be
a_k = \tau k (N -k). 
\ee
The 2-simplex pressure still remains at 
\be
\beta^\t_k =  \beta {k \choose 2}(N - k), \quad \beta > 0, k\ge 2.
\ee
Finally, for the infection pressure coming from the 3-simplices there are precisely $k \choose 3 $ ways to select 3 infected nodes and $N-k$ ways to select a susceptible node, so, after introducing the new parameter $\delta$ for the 3-simplices, we define
\be
\delta^{\square}_k = \delta {k\choose 3}(N-k), \quad \delta>0, k\ge 3.
\ee
Note that this pressure starts playing a role when $k \ge 3$, otherwise it is zero. Then the Kolmogorov equations are readily available, and we can compactly write them as 

\be
\frac{d}{dt}p_k(t) = (\alpha_{k-1} + \beta^{\triangle}_{k-1} +\delta^{\square}_{k-1})p_{k-1}(t) - (\alpha_{k} + c_k+ \beta^{\triangle}_{k} +\delta^{\square}_{k})p_k(t)+
c_{k+1}p_{k+1}(t).
\label{eq:kolm-kn3}
\ee

Following the same methodology as for Theorem~\ref{thm:hydroKn}, we formulate a the limiting mean-field equation as given below.
\begin{theorem}\label{thm:3-simplexHydro} Fix a time horizon $T>0$ and assume that the following uniform pointwise limit 
\be \label{eq:unipok3}
 \lim_{ N \to \infty} \sup_{t \le T} \Big\|\frac{1}{N} m^N(t) - m_I(t)\Big\|_{\infty} = 0
\ee
exists as a deterministic function on $[0,1]$. Furthermore, assume that $\eta^N_t$, defined by Equation~\eqref{eq:rep1}, satisfies 
\be \label{eq:ass0k3}
\lim_{N \to \infty} \sup_{t \le T} \E(\eta^N_t) = 0.
\ee
Let $\gamma$ denote the recovery rate. Define, as for the previous case of Equation~\eqref{lambda-mu} and adding another positive parameter $\theta$, the scaled pressures 
\be
\tau = \lambda N^{-1}, \quad \beta = \mu N^{-2}, \quad \delta = \theta N^{-3}
\ee
corresponding respectively to infections coming from 1-,2- and 3- simplices. Then 
\be
\lim_{N\to \infty} \frac{d}{dt}\frac{m_N(t)}{N}=\frac{d}{dt}m_I(t) = (\lambda - \gamma) m_I(t) + \Big( \frac{\mu}{2} - \lambda\Big) m_I(t)^2 +\left(\frac{\theta}{6}-\frac{\mu}{2}\right) m_I(t)^3 - \frac{\theta}{6} (m_{I}(t))^4.
\ee
\end{theorem}

A sketch of the proof is given in Appendix~\ref{App:2-simp-proof}. The finer technical details of the proof of this theorem are not included, as the arguments follow closely those of Theorem \ref{thm:hydroKn}. 

In Figure~\ref{fig:two_close_curves} (right) we show the outcome of numerical tests comparing the expected number/fraction of infectious nodes based on the Kolmogorov equations versus that resulting from solving one single ordinary differential equation, the mean-field limit. As before, the agreement between the two is excellent. The relative increase in the endemic levels in a one-to-one comparison between the left and rightmost panels reveals that the added infectious pressure brought my higher-order simplices (3-simplices) contributes less than the lower order ones (2-simplices). Furthermore, the effect of the 3-simplices becomes more/less significant when the pairwise infection rate decreases/increases.

\subsection{Higher-order $SIS$ epidemics on a complete simplicial $M$-complex}

In the previous subsections we chose to present the proof of Theorem \ref{thm:hydroKn} for a simplicial contagion up to order 2 (triangles) only for purposes and clarity. As the proof for order 3 demonstrate, the scaling performed can be generalised to higher-order simplices by building on previously obtained scaling limits. 
This means that the ingredients of the previous two proofs are sufficient to prove a general hydrodynamic result via induction, where extra infection pressure can come from any complete simplicial i-complex $1 \le i \le M$ for some fixed $M$. The formal proof is left to the reader, as it is a repetition of the previous steps and induction. The base case of the induction when $M=2$ is the proof of Theorem \ref{thm:hydroKn}.

\begin{theorem}\label{thm:FullHydro} Fix a time horizon $T>0$ and assume that the following uniform pointwise limit 
\be \label{eq:unipok}
 \lim_{ N \to \infty} \sup_{t \le T} \Big\|\frac{1}{N} m^N(t) - m_I(t)\Big\|_{\infty} = 0,
\ee
exists as a deterministic function on $[0,1]$.  Furthermore assume  $\eta^N_t$, defined by \eqref{eq:rep1}, satisfies 
\be \label{eq:ass0k}
\lim_{N \to \infty} \sup_{t \le T} \E(\eta^N_t) = 0
\ee
Let $\gamma$ denote the recovery rate. Define $r_i$ as the infection pressure on one single susceptible node in a generic active simplex of size ($i+1$), with all the other $i$ nodes being infected, $1 \le i \le M$ for some $M$ finite. Furthermore assume that each $r_i$ scales according to the order of appearance of the simplex i.e. 
\be
r_i = s_i N^{-i},  \quad k\ge 1
\ee
for some $s_i \ge 0$.

Then

\be
\lim_{N\to \infty} \frac{d}{dt}\frac{m_N(t)}{N} = \frac{d}{dt}m_I(t) = \left(s_1-\gamma \right)m_{I}(t) + \sum_{i=1}^{M-1} \left(\frac{s_{i+1}}{(i+1)!} - \frac{s_{i}}{i!} \right) (m_I(t))^{i+1}-\frac{s_{M}}{M!}(m_I(t))^{M+1}.
\ee
\end{theorem}

Note that the above also works for the pairwise infection ($M=1$), since in that case the sum is interpreted as empty and only the first and last term survives in the equation above. As expected, we have $dm_{I}(t)/dt=\left(s_1-\gamma \right)m_{I}(t)-s_1m_{I}(t)^2$.

\section{Bifurcation analysis of the resulting mean-field models}\label{sec:bif}
We now explore the phase diagrams associated to the mean-field models derived in the previous sections, for both structures.

\subsection{The complete simplicial 2-complex}  

From Theorem \ref{thm:hydroKn}, the limiting equation is
\begin{equation}
\frac{d}{dt}m_I(t) = (\lambda - \gamma) m_I(t) + \Big( \frac{\mu}{2} - \lambda\Big) m_I(t)^2 -\frac{\mu}{2} m_I(t)^3.
\label{eq:lim0_bif}
\end{equation}

The steady state $m_I=x$ is determined by this equation when $\displaystyle \frac{d}{dt}m_I(t) =0$ is substituted, i.e. $x$ is the solution of
\begin{equation}
0 = (\lambda - \gamma) x + \Big( \frac{\mu}{2} - \lambda\Big) x^2 -\frac{\mu}{2} x^3 = x \Big(  (\lambda - \gamma) +  \Big( \frac{\mu}{2} - \lambda\Big) x -\frac{\mu}{2} x^2  \Big).
\end{equation}

Assuming that the 2-simplices are actually contributing to the epidemic, that is $\mu \neq 0$, we have up to three solutions
\be \label{eq:solution_2complex}
x =0, \quad x = \frac{ \frac{\mu}{2} - \lambda  \pm \sqrt{ \Big( \frac{\mu}{2} - \lambda\Big)^2 + 2 \mu (\lambda - \gamma) }}{\mu}.
\ee

The existence of an epidemic-free state (trivial solution) does not depend on any parameter and it is always accessible. For the other solution, we can easily compute from Eq.~\eqref{eq:solution_2complex} the bifurcation diagram in the $(\lambda, \mu)$ parameter plane by using the following elementary facts about the solutions of a quadratic equation in the form $0=c+bx-ax^2$. Denoting the discriminant by $D=b^2+4ac$ and the position of the maximum by $m=\frac{b}{2a}$, the following cases can be distinguished:
\begin{enumerate}
\item If $D<0$, then there are no real solution;
\item If $D>0$ and $c>0$, then there is a positive and a negative solution;
\item If $D>0$, $c<0$ and $m>0$, then there are two positive solutions;
\item If $D>0$, $c<0$ and $m<0$, then there are two negative solutions.
\end{enumerate}
In our case, when $a=\mu/2$, $b=\mu/2-\lambda$, and $c=\lambda -\gamma$, the discriminant curve, where $D=0$ (in the positive quadrant of the parameter plane) takes the form $\lambda=-\mu/2+\sqrt{2\mu \gamma}$ and $m>0$ is equivalent to $\mu > 2 \lambda$.

Applying the simple rules above to our case leads to the bifurcation diagram shown in Figure~\ref{fig:bif_diag_pair_and_tri} (left), whose labelled regions are divided as follows:
\begin{enumerate}
\item If $\lambda < -\mu/2+\sqrt{2\mu \gamma}$, there is no non-trivial solution (domain D);
\item If $\lambda > -\mu/2+\sqrt{2\mu \gamma}$, and $\lambda > \gamma$, there is a positive and a negative solution (domain B);
\item If $\lambda > -\mu/2+\sqrt{2\mu \gamma}$, $\lambda < \gamma$ and $\mu > 2 \lambda$, there are two positive solutions (domain C);
\item If $\lambda > -\mu/2+\sqrt{2\mu \gamma}$, $\lambda < \gamma$ and $\mu < 2 \lambda$, there are two negative solutions (domain A).
\end{enumerate}

\begin{figure}[t]
     \centering
     \includegraphics[width=0.49\linewidth]{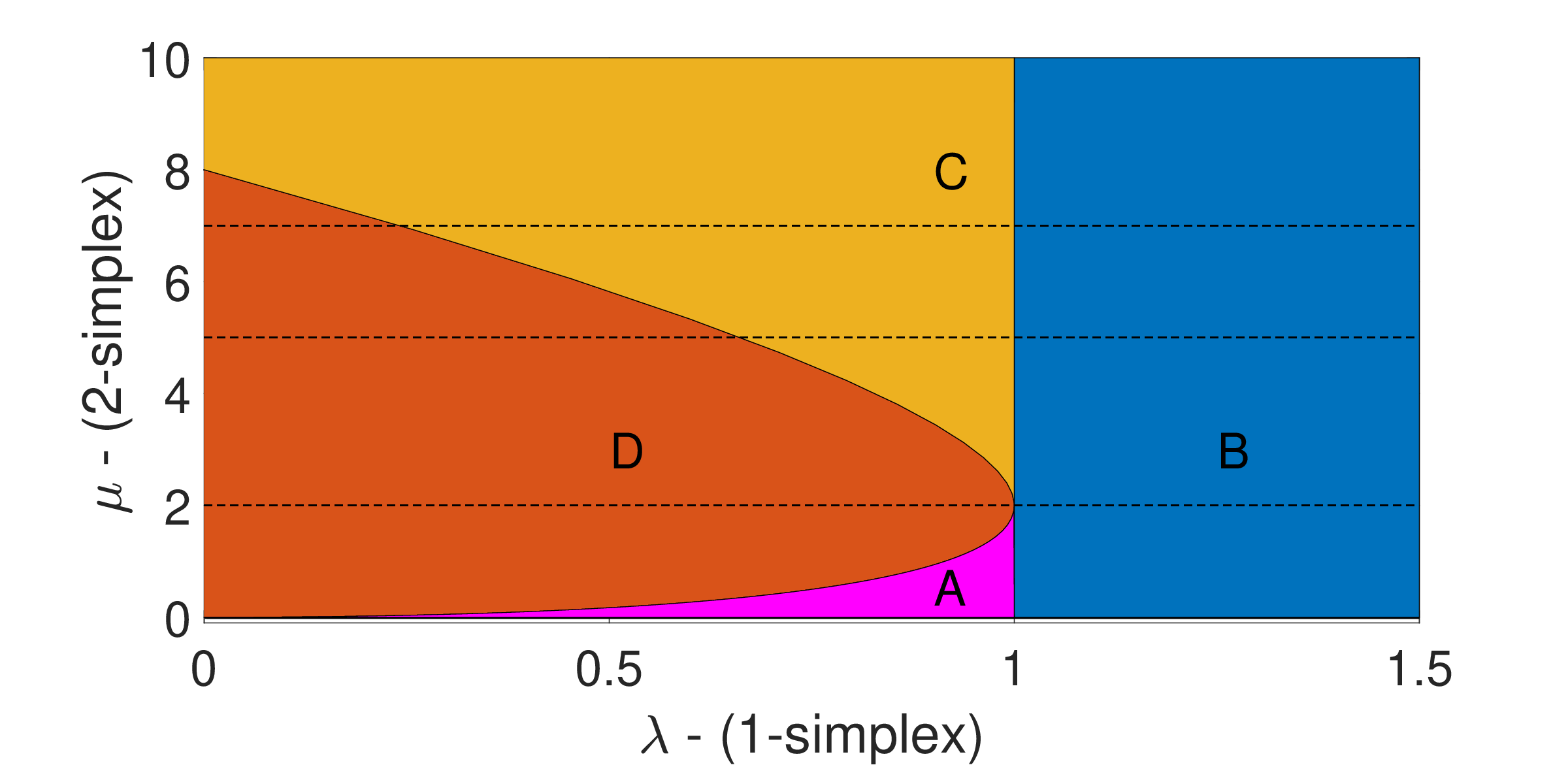}
          \includegraphics[width=0.49\linewidth]{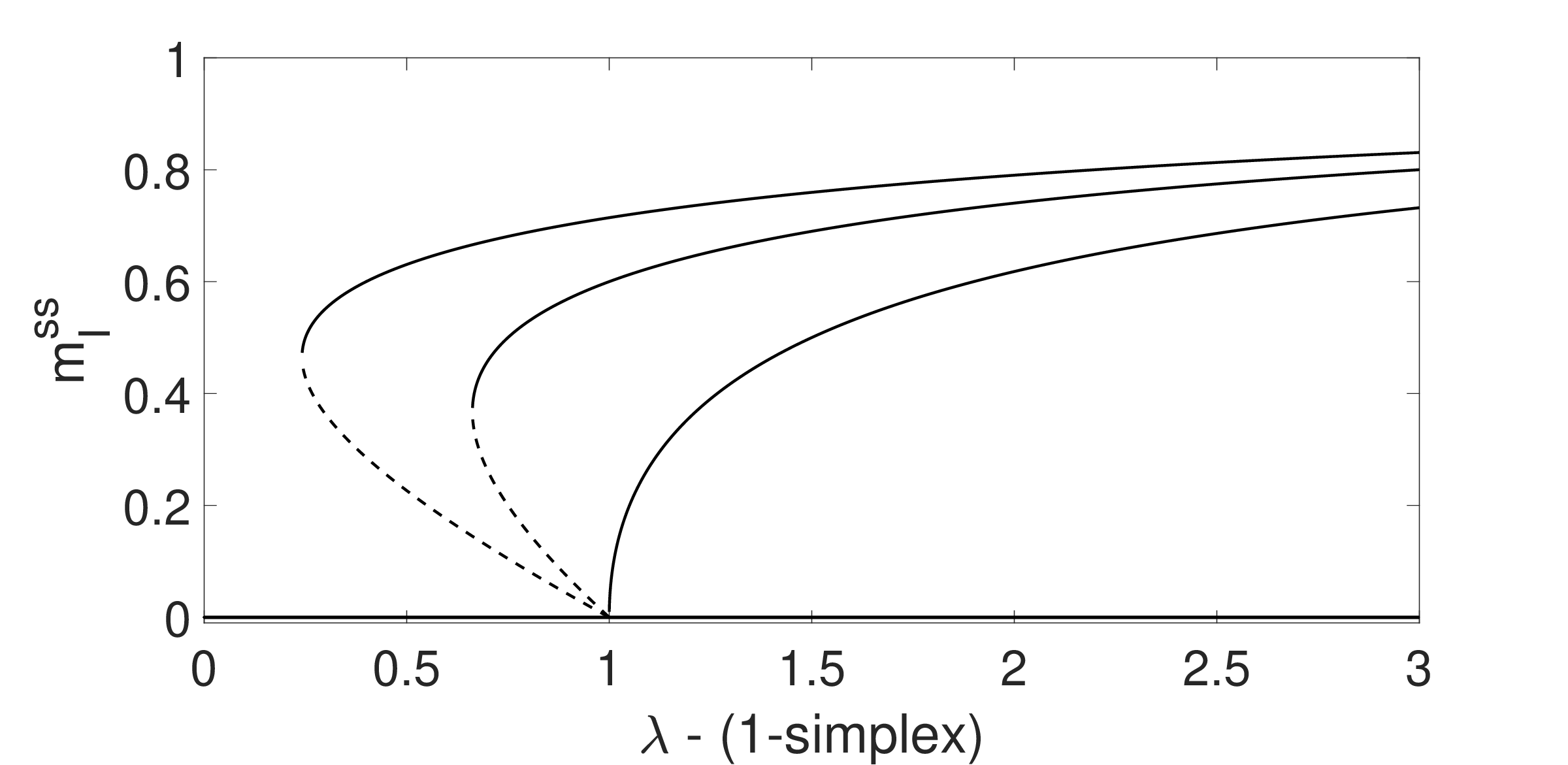}     
     \caption{
     Bifurcation diagrams for number and nature of the steady states of the differential equation in Eq.~\eqref{eq:lim0_bif} associated to a simplicial contagion dynamics on a complete simplicial 2-complex.
     Without loss of generality, we set $\gamma=1$; non-dimensionalising time  would make $\gamma$ superfluous.
     (Left) Different steady state configurations as a function of the two rescaled infectibvity parameters $\lambda$ (1-simplices) and $\mu$ (2-simplices): (A) zero and two negative solutions, (B) zero, one negative and one positive solution (the $\mu=2\lambda$ line splits this region into two further regions, above this line the magnitude of the positive solution is greater than that of the negative one, and otherwise below), (C) zero and two positive solutions, and (D) the zero solution only. (Right) Typical plot of the steady states as a function of the standard 1-simplices infectivity for two fixed values of $\mu=7,5$ and $2$ (from top to bottom). The cross sections corresponding to these values are plotted with horizontal dashed-lines in the leftmost panel.}
     \label{fig:bif_diag_pair_and_tri}
\end{figure}

The stability of these steady states can be obtained from the fact that in the cubic case, with a negative coefficient in the cubic term, the middle steady state is unstable and the other two states are stable. That is, for parameter pairs in domains A and D, $x=0$ is the only stable equilibrium. For parameter pairs in domain B, $x=0$ is unstable and there is a globally stable positive equilibrium. Finally, for parameter pairs in domain C, $x=0$ is stable and there is another stable positive equilibrium. Their basin of attraction is separated by the third, unstable steady state (also positive).

The global bifurcation picture can be sampled from the plane by ``cutting'' it at different values of $\mu$. The result is shown in Figure~\ref{fig:bif_diag_pair_and_tri} (right), where the prevalence at the steady state, denoted by $m^{SS}_I$, are plotted against the rescaled rate of infection via 1-simplices $\lambda$ for different values of the one for 2-simplices $\mu$. We note that two distinct bifurcation scenarios are possible; namely a simple transcritical bifurcation and a fold bifurcation leading to bi-stability. It is worth noting that bi-stability only appears for large values of $\mu$, i.e. a certain amount of infection mediated by 2-simplices is needed in order to sustain a region of bi-stability. We can make this more precise. In Figure~\ref{fig:bif_diag_pair_and_tri} (left) we notice that one requirement on $\mu$ for the system to display bi-stability, be in region C, when $\lambda$ is varied, is to have $\mu>2\gamma$; this is the critical $\mu_c=2\gamma$ where the parabola achieves its maximum when viewed in the ($\mu,\lambda$) plane. However, the conditions for the system to go through region C are $\lambda<\mu/2$, $\lambda<\gamma$, and $\lambda > -\mu/2+\sqrt{2\mu \gamma}$. At this point one needs to check that, if $\mu>2\gamma$ holds, then one can always find $\lambda>0$ such that the other three conditions are also satisfied. This will then guarantee that the system will go through region C when $\lambda$ is varied. One of the conditions is trivial, we can always choose a $\lambda>0$ such that $\lambda<\gamma$. However, $\lambda$ also needs to be such that $\lambda<\mu/2$. This leads to requiring that $\lambda<\min\left\{\gamma,\mu/2\right\}$. By the original assumption, $\gamma<\mu/2$, which leads to $\lambda<\gamma$. This leaves us to check that there are $\lambda$ values satisfying simultaneously the following inequaltiy: 
\be
-\frac{\mu}{2}+\sqrt{2\mu \gamma}<\lambda<\gamma.
\ee
Such values of $\lambda$ exist if and only if $-\frac{\mu}{2}+\sqrt{2\mu \gamma}<\gamma$, which is equivalent to $\sqrt{2\mu \gamma}<\gamma+\frac{\mu}{2}$. However, this is the direct consequence of the inequality between the geometric and arithmetic mean, that is,
\be
\sqrt{2\gamma\mu}<\frac{2\gamma+\mu}{2}=\gamma+\mu.
\ee

Setting $\mu = 0$ (i.e. no extra infection pressure beyond pairs) in Equation~\eqref{eq:lim0} leads to 
\be
x = 0, \quad x = 1 - \frac{\gamma}{\lambda},
\ee
which are the  steady states for the classical model with pairwise infection only.

We can use perturbation theory methods to expand the solutions of the quadratic Equation~\eqref{eq:lim0_bif}, that is the steady states, as a function of $\mu$. This effectively means that we perceive the additional infection across complete simplicial 2-complexes as a perturbation of the classic system with infections across links only. This leads to

\begin{equation}\label{eq:asympt_exp_2_simplex}
 x\simeq 1-\frac{\gamma}{\lambda}+\frac{\gamma(\lambda-\gamma)}{2\lambda^3}\mu.
\end{equation}
The zeroth-order approximation corresponds to the classical SIS model that accounts for pairwise interactions only, that is $1-\frac{\gamma}{\lambda}$. As expected, the infection across complete simplicial 2-complexes increases the value of the endemic steady  state. More precisely, from Equation~\eqref{eq:asympt_exp_2_simplex}, the contribution of the simplicial contagion is of $O(1/\lambda^2)$. Note that the expansion above only works for the case of $\lambda>\gamma$, which is the standard epidemic threshold for systems composed by pairwise interactions only. Once 2-simplices are added, the observed behaviour is usually interpreted in terms of tipping point dynamics, where depending on the initial condition the epidemic either dies out or it reaches a stable endemic equilibrium~\cite{gladwell2001tipping}.

\subsection{The complete simplicial 3-complex}

From Theorem~\ref{thm:3-simplexHydro}, the limiting equation is
\begin{equation}
 \frac{d}{dt}m_I(t) = (\lambda - \gamma) m_I(t) + \Big( \frac{\mu}{2} - \lambda\Big) m_I(t)^2 +\left(\frac{\theta}{6}-\frac{\mu}{2}\right) m_I(t)^3 - \frac{\theta}{6} m_{I}(t)^4.  
 \label{eq:3-simplex-ODE}
\end{equation}

The steady state $m_I=x$ is determined by this equation when $\displaystyle \frac{d}{dt}m_I(t) =0$ is substituted, i.e. $x$ is the solution of
\begin{equation}
0 = (\lambda - \gamma) x + \Big( \frac{\mu}{2} - \lambda\Big) x^2 +\left(\frac{\theta}{6}-\frac{\mu}{2}\right)x^3 - \frac{\theta}{6} x^4= x \Big[  \lambda - \gamma +  \left( \frac{\mu}{2} - \lambda\right) x +\left(\frac{\theta}{6}-\frac{\mu}{2}\right)x^2-\frac{\theta}{6} x^3  \Big].
\label{ss_quartic}
\end{equation}
As for the previous case, we will now determine the bifurcation diagram in the $(\lambda, \mu)$ parameter plane, while the other two parameters, $\gamma$ and $\theta$, will be fixed. The number of steady states is determined, in this cubic case as well, by the discriminant curve. We will exploit the advantages of the parametric representation method that parametrises the discriminant curve by the steady state value $x$ (see Ref.~\cite{simon1999constructing} for more details).

In order to apply the parametric representation method, the cubic term in Equation~\eqref{ss_quartic} is written in the form
\begin{equation}
f_0(x)+ \lambda f_1(x) + \mu f_2(x) =0 ,
\label{ss_cubic}
\end{equation}
where
\begin{equation}
f_0(x)= \frac{\theta}{6} (x^2-x^3)-\gamma, \qquad f_1(x)= 1-x, \qquad f_2(x)= \frac12 (x-x^2) .
\label{f0f1f2}
\end{equation}
The discriminant consist of those parameter pairs, where the function has a double root, that is its derivative is also zero at the root, i.e.
\begin{equation}
f_0'(x)+ \lambda f_1'(x) + \mu f_2'(x) =0
\label{ss_cubic_prime}
\end{equation}
holds as well. The main idea of the parametric representation method is to solve system given by Equations \eqref{ss_cubic}-\eqref{ss_cubic_prime} for $\lambda$ and $\mu$ in terms of $x$, leading to the parametric expression of the discriminant curve. This is especially useful, when the parameters are involved linearly, as in our case. Then the solution of Equations \eqref{ss_cubic}-\eqref{ss_cubic_prime} can be easily given as
\begin{equation}
\lambda = \frac{f_0'(x) f_2(x)-f_0(x)f_2'(x)}{f_1(x)f_2'(x) - f_1'(x) f_2(x)} \ , \qquad \mu = \frac{f_0(x) f_1'(x)-f_0'(x)f_1(x)}{f_1(x)f_2'(x) - f_1'(x) f_2(x)} \ .
\label{discriminant_curve_gen}
\end{equation}
After substituting Equation~\eqref{f0f1f2} into Equation~\eqref{discriminant_curve_gen}, the discriminant curve is obtained as
\begin{equation}
\lambda = \frac{\theta}{6} x^2 + \gamma \frac{1-2x}{(1-x)^2} \ , \qquad \mu = \frac{2\gamma}{(1-x)^2} -  \frac{\theta}{3} x \ .
\label{discriminant_curve}
\end{equation}
The advantage of the parametric representation method is evident since eliminating $x$ from these two equations is quite complicated, i.e. to derive the equation of the discriminant curve without parametrising with $x$ would be difficult.

We now numerically plot the discriminant curve while varying the value of $x$, as it is shown in Figure~\ref{fig:bif_theta_2_10} for two fixed values of the highest-order infectivity parameters $\theta = 2$ (left) and $\theta = 10$ (right) ---modulating the effects of the 3-simplices. The parameter $x$ of the curve varies along the real axis. We divide the curve into three parts, shown with different colours, as follows. The curve tends to infinity when $x=1$, hence we will consider the part belonging to $x<1$ and $x>1$ separately. Moreover, we are interested in positive steady states, hence we will divide the the $x<1$ part into two parts, belonging to $x<0$ and to $0<x<1$. The parameters $\lambda$ and $\mu$ are positive, hence we will investigate only the positive quadrant of the $(\lambda, \mu)$ parameter plane. The following observation can be made based on the shown results, and can be also proved by elementary calculations. 

\begin{proposition}
The following statements hold for the discriminant curve given in Equation~\eqref{discriminant_curve}.
\begin{enumerate}
\item The part belonging to $x>1$ does not enter the positive quadrant, hence we will not consider it in further investigations.
\item The curve touches the vertical line $\lambda = \gamma$ at $\mu=2\gamma$ where $x=0$.
\item The curve is locally on the left-hand-side of the vertical line $\lambda = \gamma$, if $\theta < 6\gamma$, and it lies on the right-hand-side of this vertical line, when $\theta > 6\gamma$.
\end{enumerate}
\end{proposition}
We note that the last statement follows easily when the formula for $\lambda$ in Equation~\eqref{discriminant_curve} is rearranged as
\be
\lambda = \gamma + x^2 \left( \frac{\theta}{6}  -  \frac{\gamma}{(1-x)^2} \right) .
\ee
Figure~\ref{fig:bif_theta_2_10} also shows that the discriminant curve may have a cusp point. By definition, the cusp point of a curve is that point where $\lambda'(x)=0$ and $\mu'(x)=0$ hold at the same time. The parametric representation offers an easy way to determine the cusp point, namely, the equation
\begin{equation}
f_0''(x)+ \lambda f_1''(x) + \mu f_2''(x) =0
\label{ss_cubic_doubleprime}
\end{equation}
holds as well at the cusp point. Solving the system of Equations~\eqref{ss_cubic}-\eqref{ss_cubic_prime}, \eqref{ss_cubic_doubleprime} for $x$  and substituting Equation~\eqref{f0f1f2} leads to the following equation for the $x_c$ parameter value of the cusp point:
\be
\theta (1-x_c)^3= 6\gamma .
\ee
This equation has a positive solution for $x_c$ if and only if $\theta > 6\gamma$. Hence, we have the following proposition.

\begin{proposition}
The following statements hold for the discriminant curve given in Equation~\eqref{discriminant_curve}.
\begin{enumerate}
\item The branch of the discriminant curve belonging to $x>0$ has a cusp point, if $\theta > 6\gamma$.
\item The branch of the discriminant curve belonging to $x>0$ is a convex arc, if $\theta < 6\gamma$.
\end{enumerate}
\end{proposition}

\begin{figure}[t]
     \centering
 \includegraphics[width=0.49\linewidth]{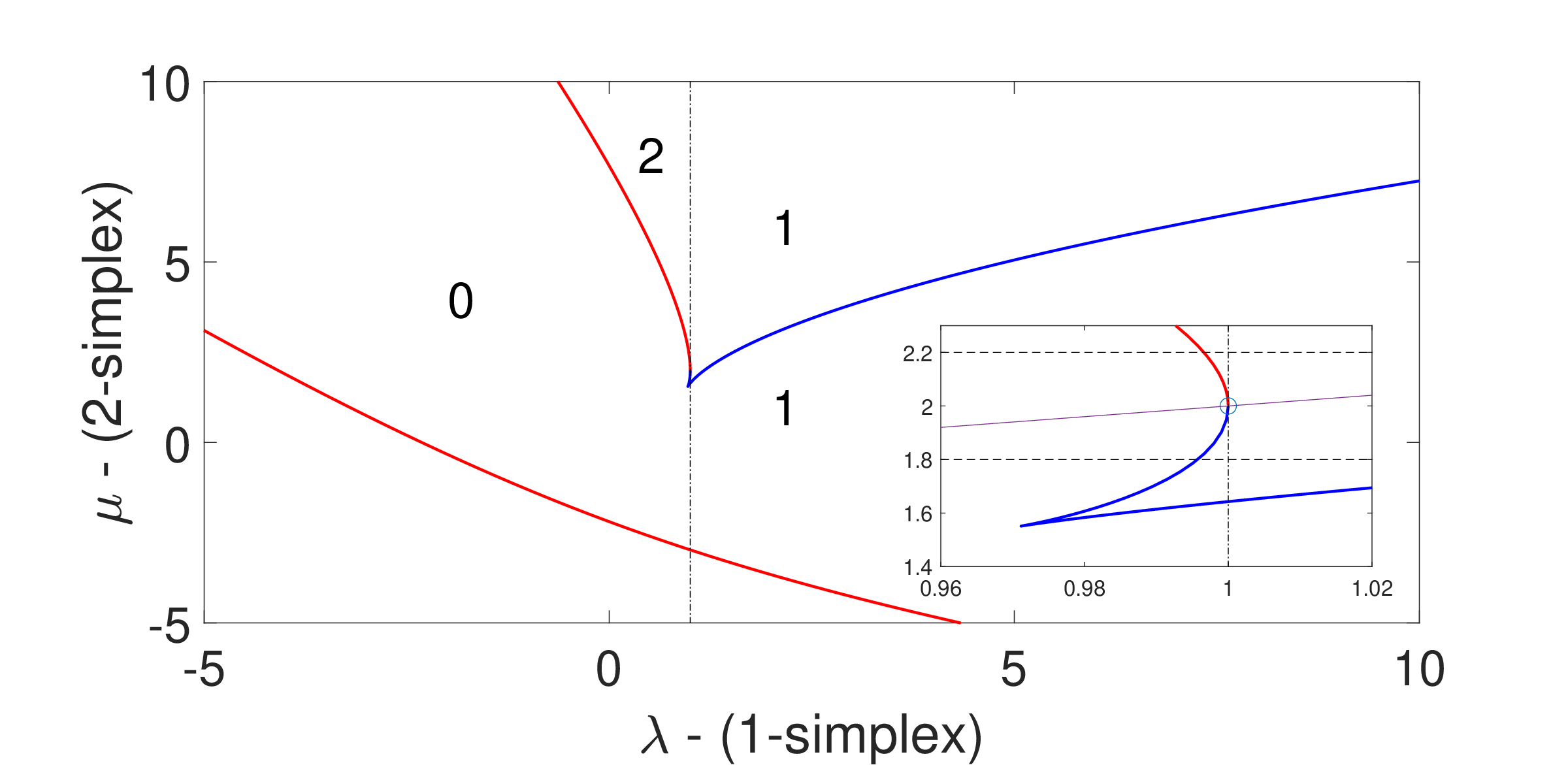}
 \includegraphics[width=0.49\linewidth]{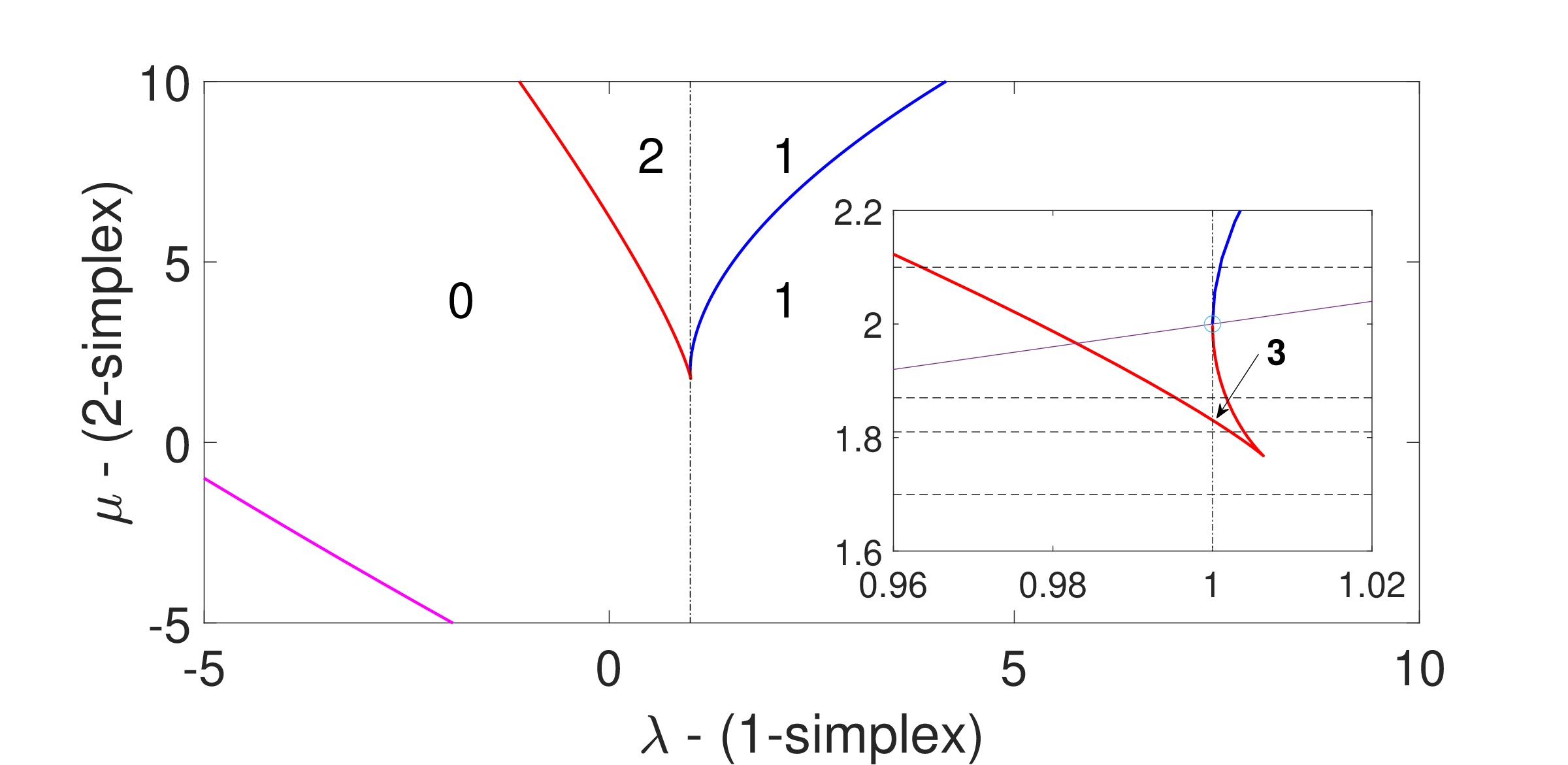} 
     \caption{ The full bifurcation analysis for a simplicial contagion dynamics on a complete simplicial 3-complex showing the number of solutions in the ($\lambda,\mu$) plane for two fixed values of $\theta=2$ (left) and $\theta=10$ (right); the numbers do not include the trivial disease-free steady state.
     Without loss of generality, we set $\gamma=1$; non-dimensionalising time  would make $\gamma$ superfluous.
     The magenta, red and blue lines correspond to $x>1$, $0<x<1$ and $x<0$, respectively. The insets show the subtle geometry of the cusps; these shift from the negative to the positive branch and thus increasing the number of possible steady states. The dot-dashed lines correspond to the $\lambda=\gamma$ line and the oblique line in the inset corresponds to the $\mu=2\lambda$ line, these are important to determine when three strictly positive steady states exist. Finally, the dashed lines in the insets represent fixed values of $\mu$ (1.8 and 2.2 in the leftmost panel, and 1.7, 1.81, 1.87 and 2.1 in the rightmost panel) where the bifurcation profile changes when $\lambda$ is varied; these are shown explicitly in Figures~\ref{fig:bif_theta_2_cross_section} and~\ref{fig:bif_theta_10_cross_section}. The marker ($\circ$) in the insets is the transition point from the positive to the negative branch. In the inset of the rightmost panel, the domain with three solutions, or four if zero is included, is the region bounded by the cusp and the $\lambda=\gamma=1$ line.}
        \label{fig:bif_theta_2_10}
\end{figure}

Finally, in order to construct the bifurcation diagram according to the number of positive steady states, we will make use of the so-called tangential property of the parametric representation method (see Ref.~\cite{simon1999constructing} for more details).

\begin{enumerate}
\item The number of solutions of Equation~\eqref{ss_cubic} for a given parameter pair $(\lambda, \mu)$ is equal to the number of tangents that can be drawn from $(\lambda, \mu)$ to the discriminant curve.
\item The values of solutions of Equation~\eqref{ss_cubic} for a given parameter pair $(\lambda, \mu)$ are the $x$ parameter values of the tangent points along the discriminant curve.
\end{enumerate}

The number of tangents that can be drawn to a convex arc can be easily read geometrically. It can be 2, 1 or 0 according to the position of the point from which we draw the tangent lines. For a cusp however, from certain points it is possible to draw three tangents and it is this property which allows to further increase the number of steady states.

Summarising, we can have two different bifurcation diagrams according to the specific values of the parameters $\theta$ and $\gamma$, as it is shown in Figure~\ref{fig:bif_theta_2_10}. If $\theta < 6\gamma$ (leftmost panel), the cusp of the bifurcation curve lies in the region where $x<0$, thus giving no biologically meaningful steady states. Note that, besides the number of solutions shown in the figure, the disease-free steady state is always a steady state. Contrarily, in the case $\theta > 6\gamma$ (rightmost panel), the cusp is on the $x>0$ branch, and thus the number of positive steady states can be 0, 1, 2 or 3. 

The problem of finding the number of positive solutions of Equation~\eqref{eq:3-simplex-ODE} can also be approached by using Descartes’s rule of signs. This states that the number of positive solutions is related to the sign changes in the coefficients of the polynomial when arranged in the canonical order. Since the coefficient of the highest-order term, $-\theta/6$, is negative, we must have that 
\be
\frac{\theta}{6}-\frac{\mu}{2}>0,\,\,\, \frac{\mu}{2}-\lambda<0, \,\,\,\lambda-\gamma>0 \Longleftrightarrow \theta>3\mu,\,\,\, \mu<2\lambda,\,\,\,\lambda>\gamma.
\ee

\begin{figure}[t]
     \centering     \includegraphics[width=0.49\textwidth]{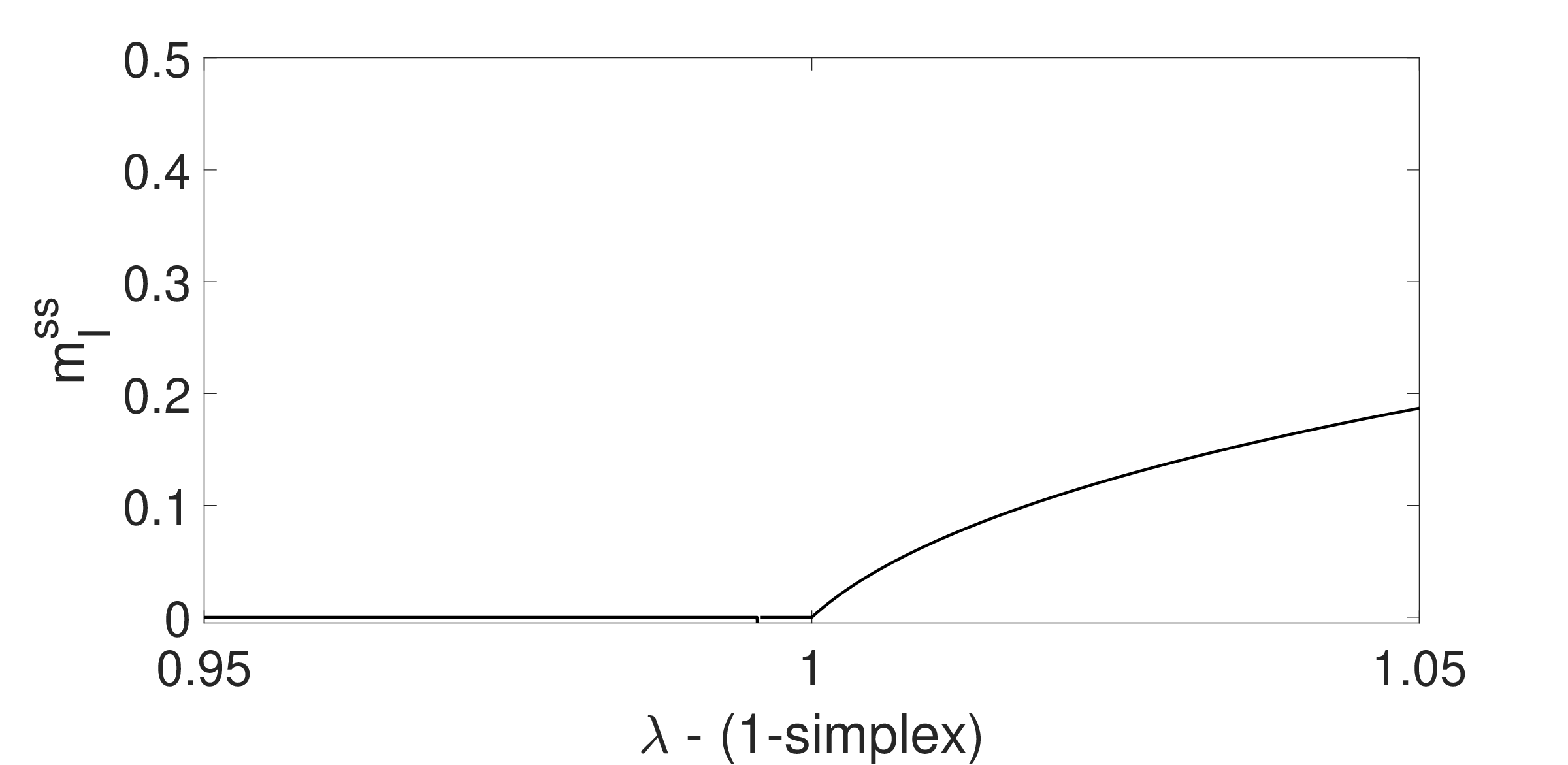}   
 \includegraphics[width=0.49\textwidth]{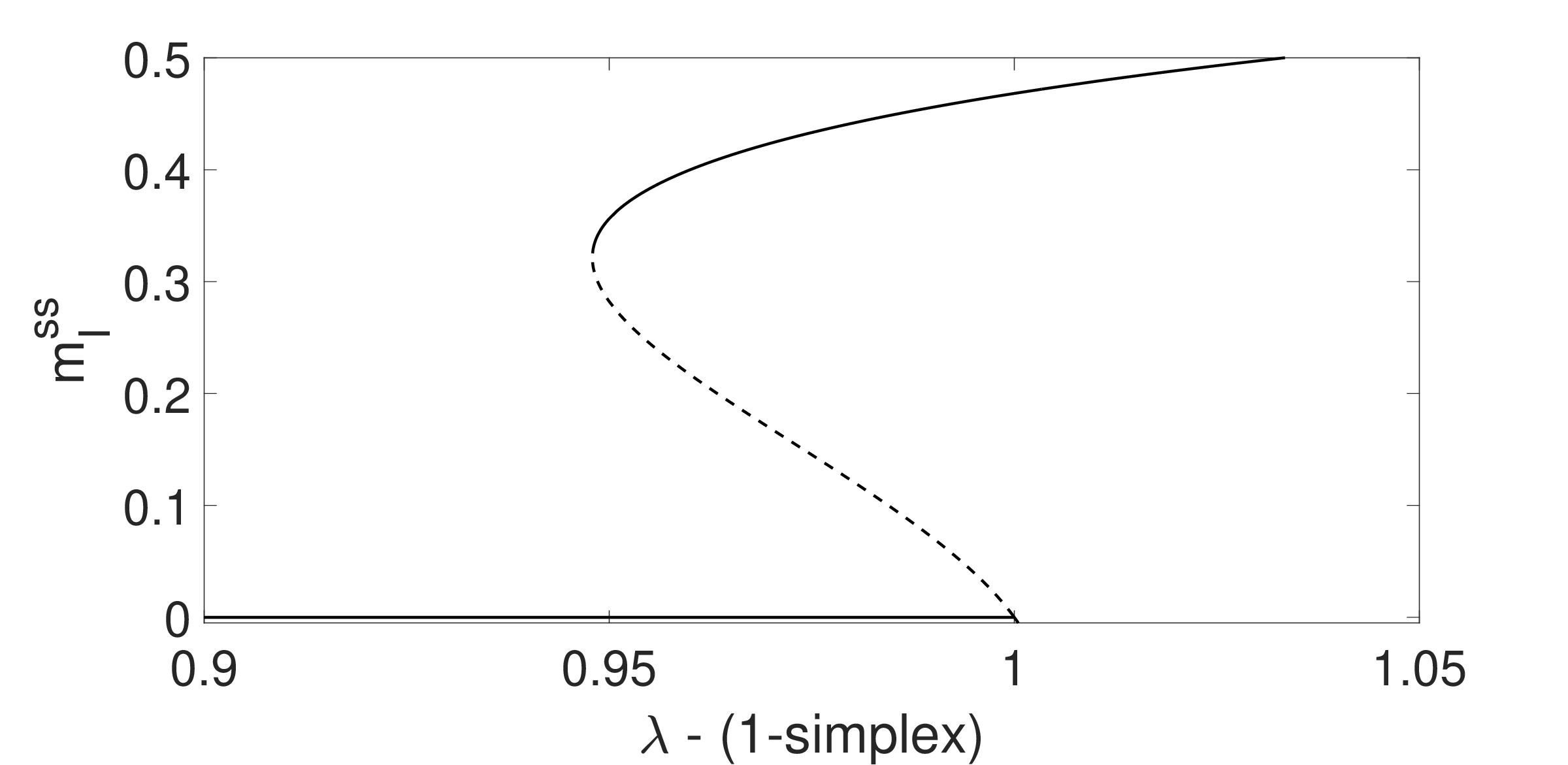} 
     \caption{Bifurcation picture in the case $\theta=2$ for $\mu=1.8$ (left) and $\mu=2.2$ (right). Increasing the value of $\mu$, for this particular value of $\theta$, moves the system from a simple transcritical bifurcation to fold bifurcation leading to bi-stability.}
        \label{fig:bif_theta_2_cross_section}
\end{figure}

It is worth to plot the steady states as $\lambda$ is varied for different fixed value of $\mu$. 
The inset in the leftmost panel of Figure~\ref{fig:bif_theta_2_10} shows that there are two regions of particular interest; namely $\mu=1.8$ and $\mu=2.2$. As shown in Figure~\ref{fig:bif_theta_2_cross_section}, the system goes from displaying a simple transcritical bifurcation ($\mu=1.8$) to a fold bifurcation leading to bi-stability ($\mu=2.2$). In Figure~\ref{fig:bif_theta_10_cross_section} we consider instead higher contributions coming from the highest-order simplices. Setting $\theta=10$ leads to a more complex scenario where the system can display the transcritical behaviour, fold bifurcation with the fold after the transcritical point, fold bifurcation with the fold before the transcritical point, and bi-stability between an endemic and the disease-free steady state. It may be useful to focus on the two fold bifurcations and in particular on the cross-sections shown in the inset of the rightmost panel of Figure~\ref{fig:bif_theta_2_10}. From here on, the number of solutions in a region includes the trivial disease-free solution. For $\mu=1.81$ (Figure~\ref{fig:bif_theta_2_10}, top-right panel), and with $\lambda$ moving from zero to positive values, it is evident that for small values of $\lambda$ the only steady state is the trivial disease-free steady state. However, as $\lambda$ increases the first transition point is into the area with two steady states (transcritical point), this is the narrow region between the $\lambda=1$ line and the left-hand side of the cusp. Increasing $\lambda$ further takes the system to the second transition point inside the cusp with four solutions and finally back to a region with two solutions. For $\mu=1.87$ (Figure~\ref{fig:bif_theta_2_10}, middle-left panel), as $\lambda$ increases from zero, we move from one solution to three (point of the fold) and then to four solution as we pass through the transcritical point. As $\lambda$ increases further, only the stable endemic and unstable trivial disease-free states survive. In this case, on entering the cusp from the left we have three solutions and this increase to four as the system is still within the cusp but passes through the $\lambda=\gamma=1$ boundary.

To further illustrate these findings, we plot the temporal evolution of the prevalence as given by the solutions of Equation~\eqref{eq:3-simplex-ODE} starting from different initial conditions, at the bottom panel of Figure~\ref{fig:bif_theta_10_cross_section}. The figure clearly shows that we have two stable non-zero steady states, as expected based on the top-right panel in the same figure. Hence, the long-term behaviour of the system strongly depends on the initial conditions.
The stability of the steady states can be obtained by using the fact that the largest steady state is stable. Hence, for example, in the case of 3 positive steady state, the largest and smallest positive steady states are stable, while the middle one and zero are unstable.

\begin{figure}[h!]
     \centering 
\includegraphics[width=0.49\textwidth]{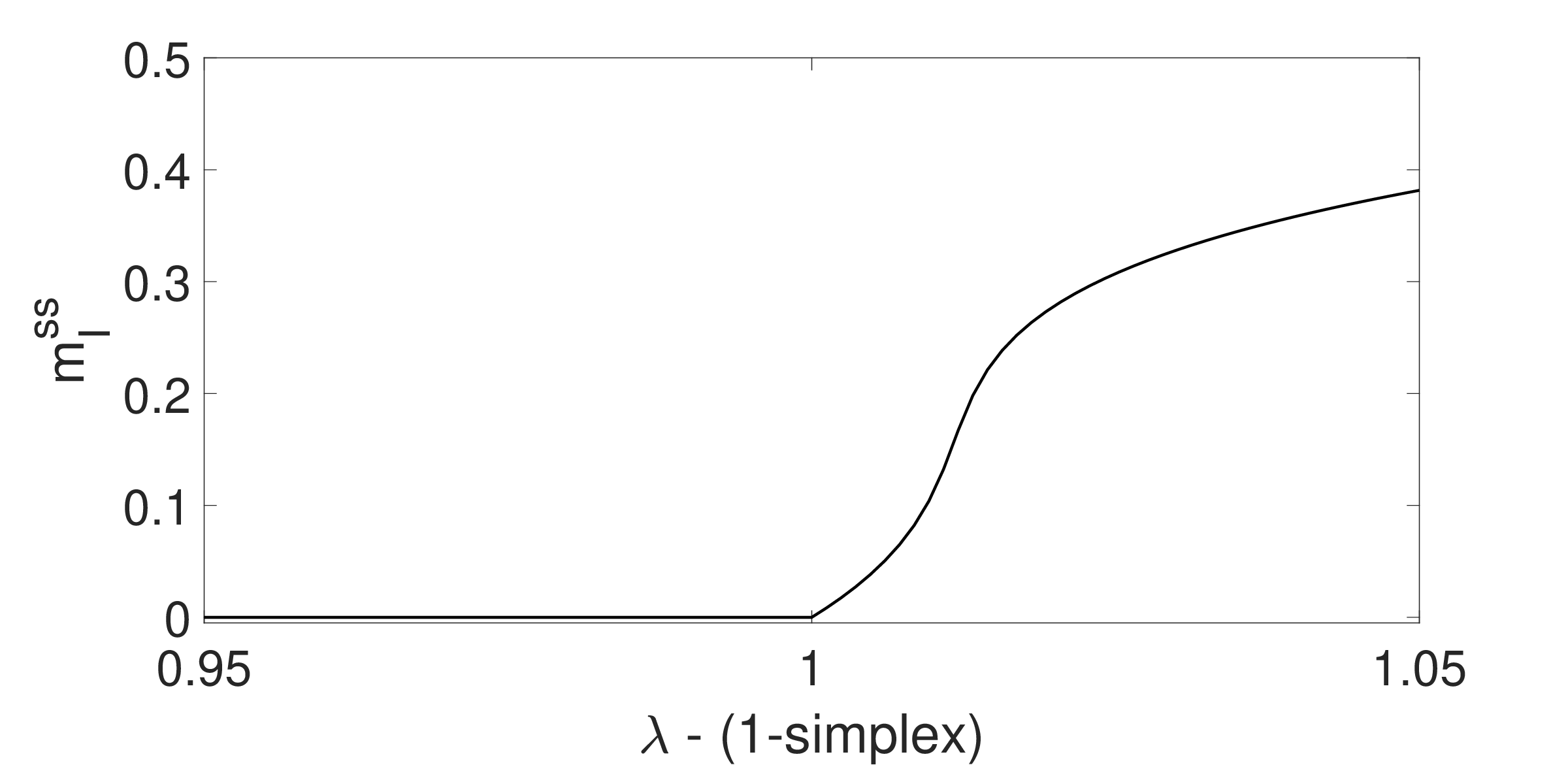}
\includegraphics[width=0.49\textwidth]{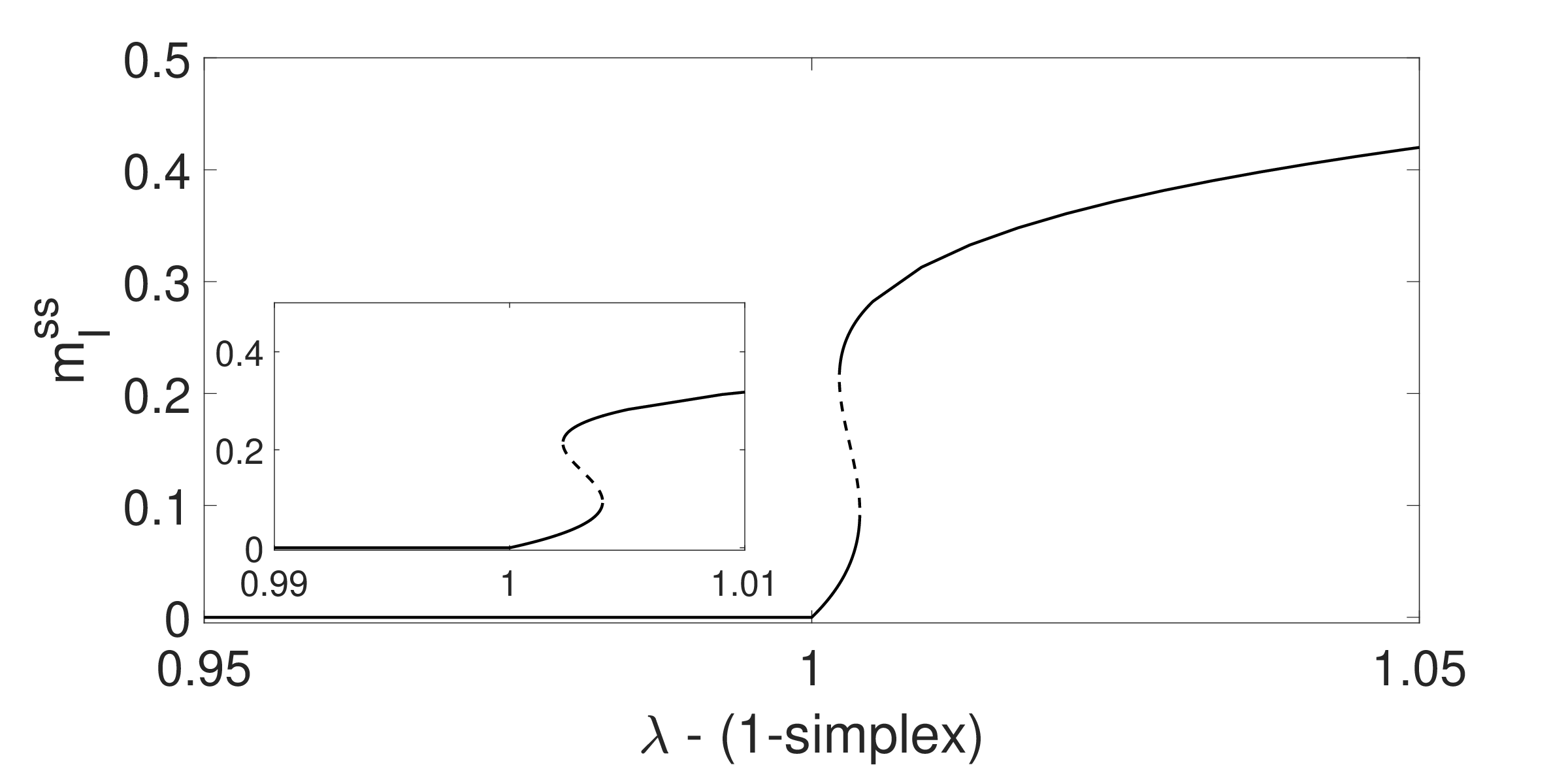} 
\includegraphics[width=0.49\textwidth]{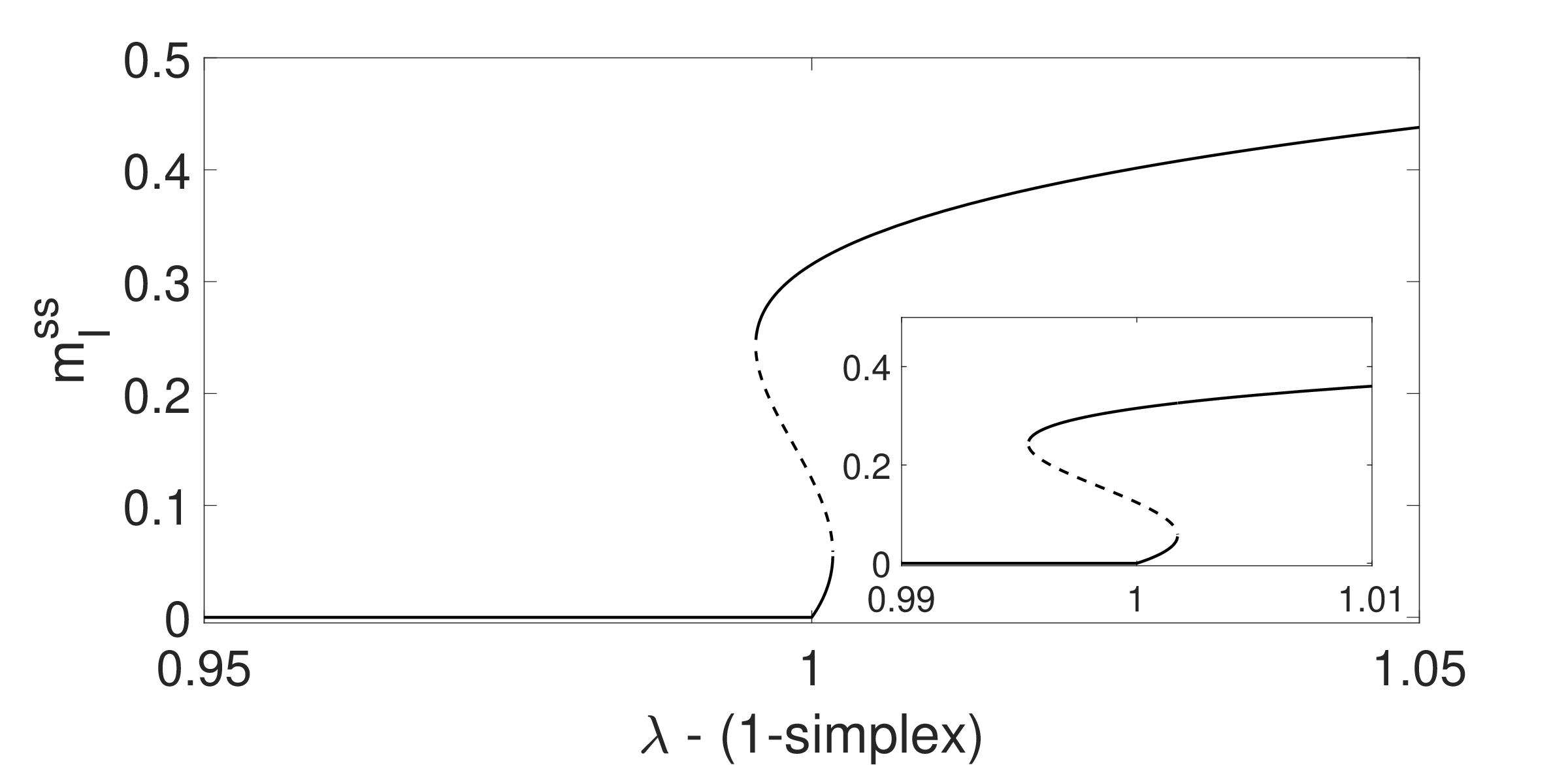}
\includegraphics[width=0.49\textwidth]{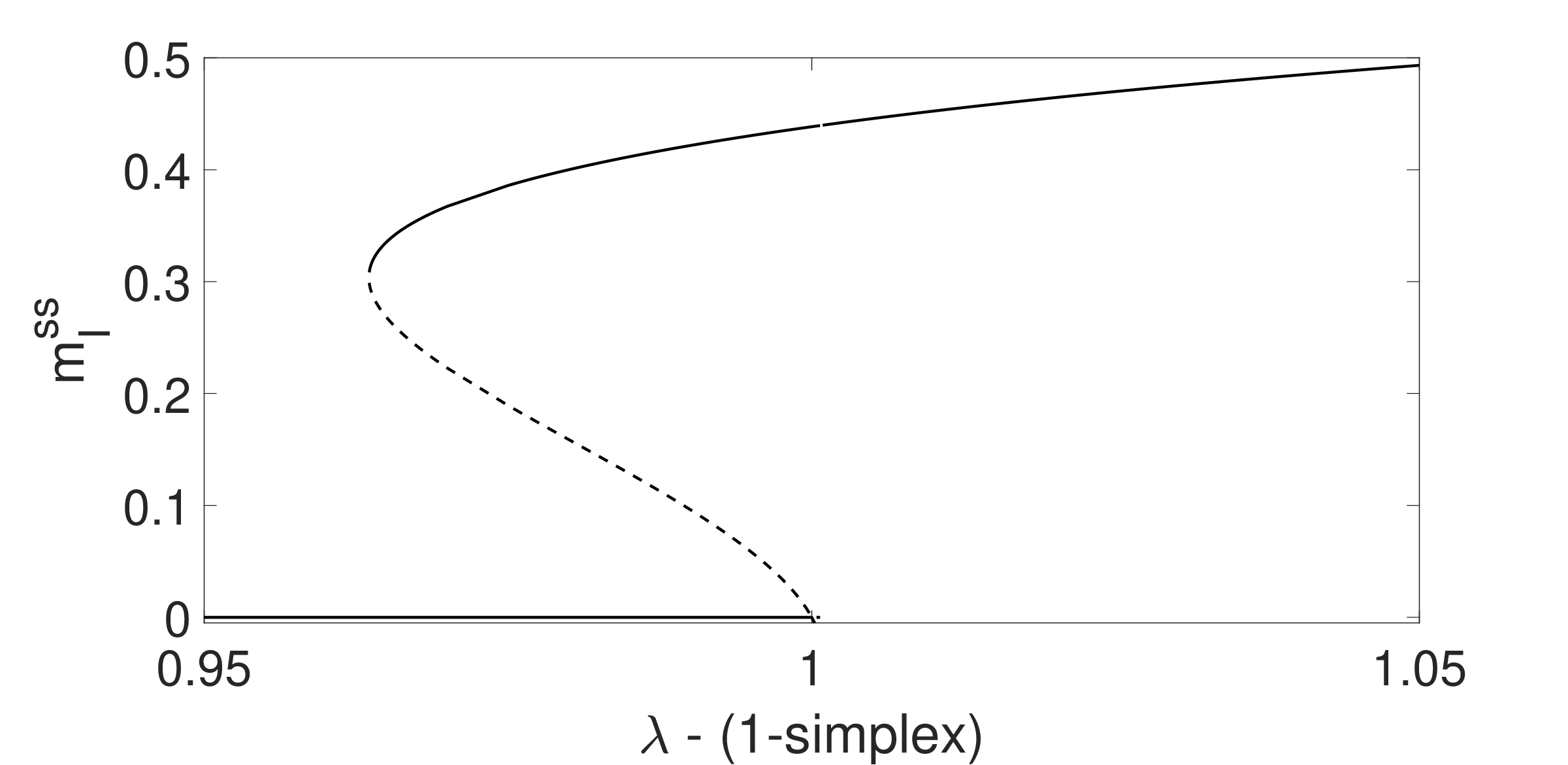}
\includegraphics[scale=0.25]{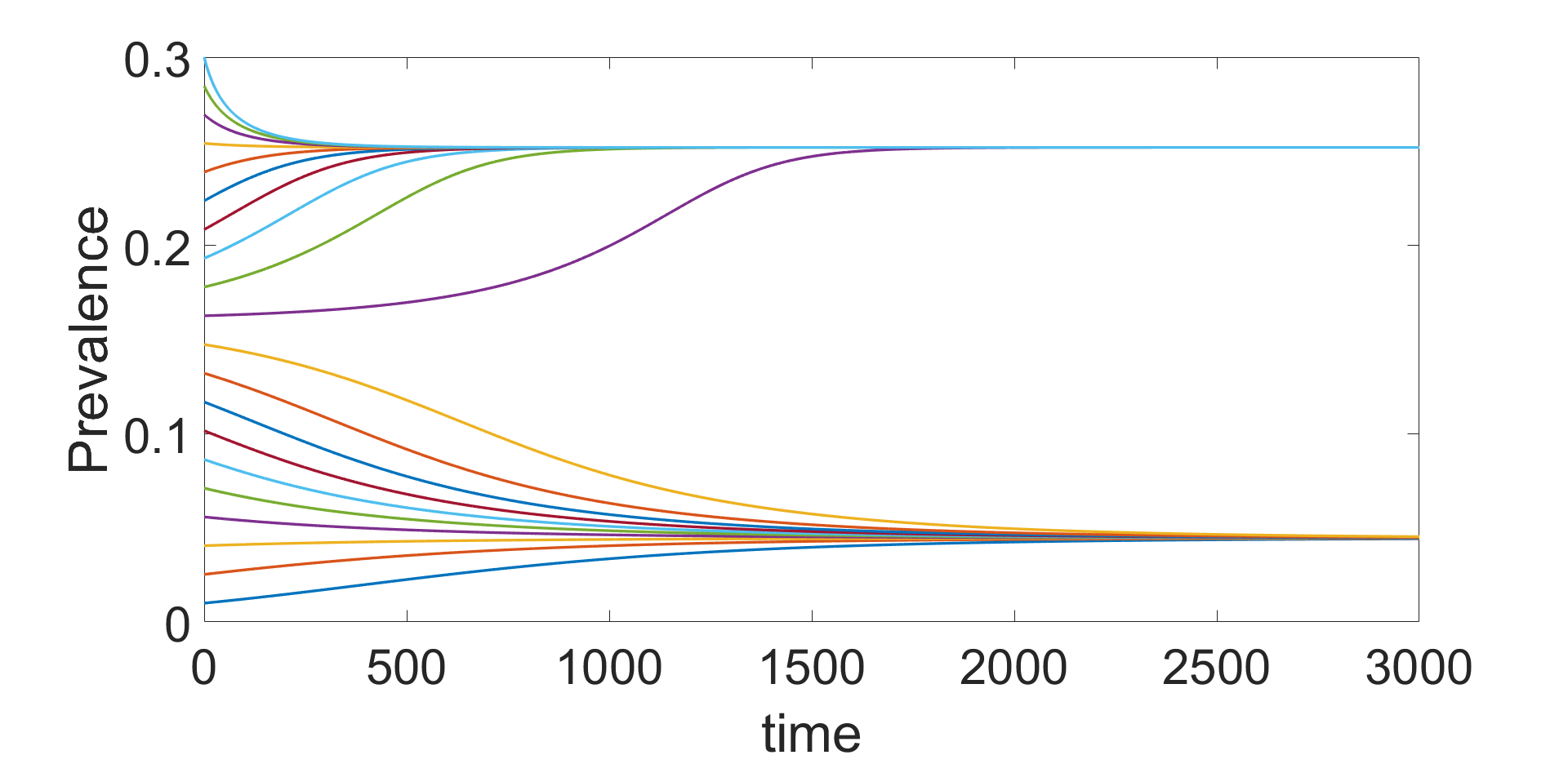}
    \caption{Bifurcation picture in the case $\theta=10$ for $\mu=1.7$ (top-left), $\mu=1.81$ (top-right), $\mu=1.87$ (middle-left) and $\mu=2.1$ (middle-right). Increasing the value of $\mu$, for this particular value of $\theta$, moves the system through the following four distinct bifurcation profiles: transcritical, fold bifurcation with fold after the transcritical point, fold bifurcation with fold before the transcritical point and bi-stability between an endemic and the disease-free steady state. The figure at the bottom shows solutions of the mean-field model, with up to 3-simplices. Each individual solution starts with a different initial state and it illustrates that the system has two stable steady states. Parameters are: $\gamma=1$, $\lambda=1.003$, $\mu=1.81$ and $\theta=10$. This plot corresponds to fixing $\lambda=1.003$ in the top-right plot.}
        \label{fig:bif_theta_10_cross_section}
\end{figure}

Finally, it is possible to obtain an asymptotic expansion to one of the solutions of equation~\eqref{eq:3-simplex-ODE} in the limit of small $\mu$ and $\theta$. This leads to 
\begin{equation}
 x\simeq 1-\frac{\gamma}{\lambda}+\frac{\gamma(\lambda-\gamma)}{2\lambda^3}\mu+\frac{\gamma(\lambda-\gamma)^2}{6\lambda^4}\theta.
\end{equation}
The expansion above shows again that the added infection due to complete simplicial 3-complex increase the value of the endemic steady sate. Under certain conditions we can again see that this is of order $O(
1/\lambda^2)$ which again highlights how high values of the rate of infection across pairs limits the effect of higher-order structures.

\section{Discussion}\label{sec:disc}
Understanding how behavioural contagion unfolds on a population of interacting individuals is sociologically interesting, but also crucial for biological spreading given its strict interplay with the behavioural component that can facilitate/inhibit the contagion process~\cite{perra2011towards, scarpino2016effect, perra2021non, lucas2023simplicially}. In this paper, we derived the exact equations for an $SIS$-like simplicial contagion dynamics on fully connected simplicial complexes where contagion events are mediated by $k$-simplices of arbitrary order that represent group interactions from 2-body to ($k+1$)-body. While the exact model can be written out explicitly and can be evaluated numerically, we also provided a rigorous mean-field limit in the form of a single differential equation for the expected number of infected nodes. We then performed a detailed bifurcation analysis for the case of a simplicial contagion that involves complete simplicial 2- and 3-complexes. In both cases, we found that the novel effects brought by higher-order interactions can be effectively interpreted as as perturbations to the base case of an SIS spreading only through pairs of nodes.
We analytically showed how the higher-order structure contributes to increasing the value of the endemic steady state. In particular, we found that this contribution is of $O(1/\lambda^2)$, where $\lambda$ is the rate of infection across a unique link. This clearly illustrates the vanishing contribution of simplicial contagion as the value of the pairwise infection rate increases.

Our work is the first one to provide an analytical treatment that is able to handle simplicial contagion that runs on simplices of order higher than 2. In fact, most of the recent studies in the emerging field of ``higher-order network science'' have been focusing on understanding the dynamical differences that emerge when the system descriptors move beyond pairs; in this view, triangles, encoding 3-body interactions, are the most natural starting point. In Ref.~\cite{iacopini2019simplicial}, the particular case of a simplicial contagion running exclusively on 1-simplices and simplices of another single higher order $k>1$ was analysed, leading to similar results as for the simplest case of 1- and 2-simplices. In this manuscript, we made an important stride in this direction by allowing for all interactions at multiple orders up to 4 nodes, and also providing a general equation for infections running on systems featuring all interactions up to an arbitrary fixed order.

We found that there is a natural relation between the order of the simplex and the degree of the polynomial appearing on the right-hand side of the mean-field limit. For pairs only, the mean-field equation is driven by a quadratic polynomial, for up to complete simplicial 2-complexes is cubic, for complete simplicial 3-complexes is quartic and for complete simplicial $M$-complexes the polynomial has degree ($M+1$). This naturally leads to the possibility of having multiple non-trivial (excluding the trivial disease-free) steady states, and it comes down to the number of positive solutions of polynomials. Indeed, as for previous results~\cite{iacopini2019simplicial, matamalas2020abrupt, de2020social, landry2020effect, barrat2022social}, when we consider a structure that allows for up to 2-simplices we find a bi-stability regime where a stable endemic and a stable disease-free states co-exist. Here, we showed that when allowing for 3-simplices as well it is possible to have four steady states, two of which are distinct strictly positive endemic steady states. Given these premises, we envision that going even higher with the order of the interactions could lead to having an arbitrary numbers of strictly positive and distinct endemic steady states.
Note that multistability was already found in previous threshold-based contagion models on higher-order networks that presented community structure. Here, instead, we show that even the most trivial structure, a fully connected one, can lead to this phenomenon if one simply allows for interaction of higher-order. It is well known that the outcome of an epidemic process is the result of the complex interplay between the structure of the underlying network and epidemic dynamics that unfolds on top of it. Since the models analysed in this manuscript are all based on the assumption of a fully connected structure, we conclude that in this case the richness of model behaviour is mainly driven by the formulation of the dynamics through higher-order mechanisms. 

We foresee different possible extensions to our work. The most natural continuation would be to investigate the possibility of making general statements about the global bifurcation picture of the mean-field model given by polynomials of arbitrary order, coupled with an extension of the results obtained here via perturbation theory.
Next, one could make crucial realistic steps in two directions.
On one hand, it would be interesting to rigorously explore the effects of simplicial contagion of arbitrary order on higher-order structures that are not complete. Preliminary calculations on \ER structures show that it is possible to derive a semi-rigorous Markov-chain at the population level (not to be confused with the microscopic-Markov chain approach~\cite{matamalas2020abrupt, burgio2021network, burgio2023triadic}) where infection rates involving 1- and 2-simplices are approximated based on probabilistic arguments.
On the other hand, it is well known that temporality of contacts can have a huge impact on the spreading of epidemics on complex networks~\cite{holme2012temporal}, and higher-order structures can obviously present (even more) interesting temporal patterns at all orders~\cite{cencetti2021temporal, gallo2023higher, iacopini2023temporal}.
The problem of understanding the impact of having group interactions that change in time on the dynamical processes that unfold on top has been addressed only by a few studies ---based on simulations~\cite{chowdhary2021simplicial, neuhauser2021consensus}. Is it possible to devise a formal analytical treatment for (even trivial) time-varying higher-order structures? Starting from the simplest approach possible, one could keep the structure as it is (fully connected), but instead investigate the impact of having group events that get activated at different points in time according to some activity distribution. 


\appendix

\section{}
\label{App:2-simp-proof}
In this section, we detail the proof for deriving the mean-field limits for the complete simplicial 2- and 3-complexes.
\subsection{Mean-field limit for complete simplicial 2-complex}
\begin{proof}[Proof of Theorem \ref{thm:hydroKn}.]
The expected number of infected $m^N(t) = \E(I^N_t) = \sum_{k=0}^{N}kp_k(t)$ at time $t$ satisfies the equations 

\begin{align}
\frac{d}{dt} m^N(t) & = \sum_{k= 0}^N k \frac{d}{dt} p_k(t) \notag \\ 
& = \sum_{k=0}^N k [a_{k-1}p_{k-1}(t) - (a_k + c_k)p_k(t) + c_{k+1}p_{k+1}(t)] +  \sum_{k=0}^N k \beta^\t_{k-1}p_{k-1}(t)   -  \sum_{k=0}^N k\beta^\t_k p_{k}(t). \label{eq:mean-in}
\end{align} 

Focus on the middle sum in \eqref{eq:mean-in}. By the initial conditions on $\beta_k^\t$ we have 
\begin{align*}
 \sum_{k=0}^N k \beta^\t_{k-1}p_{k-1}(t) &= \sum_{k=3}^{N} k \beta^\t_{k-1}p_{k-1}(t) \\
 &= \sum_{\ell = 2}^{N-1} (\ell+1) \beta^\t_{\ell}p_{\ell}(t) =  \sum_{\ell = 2}^{N-1} \ell \beta^\t_{\ell}p_{\ell}(t) +  \sum_{\ell = 2}^{N-1} \beta^\t_{\ell}p_{\ell}(t)\\
 &= \sum_{\ell = 0}^{N} \ell \beta^\t_{\ell}p_{\ell}(t) +  \sum_{\ell = 0}^{N} \beta^\t_{\ell}p_{\ell}(t), \quad \text{since $\beta_N^\t = 0$ as well.}
\end{align*}

Substitute back in \eqref{eq:mean-in} to obtain 
\be\label{eq:ready2hydro}
\frac{d}{dt} m^N(t) =  \sum_{k=0}^N k [a_{k-1}p_{k-1}(t) - (a_k + c_k)p_k(t) + c_{k+1}p_{k+1}(t)] + \beta \sum_{\ell = 2}^{N-1} (N -\ell) {\ell \choose 2} p_{\ell}(t).
\ee

We build on Eq.~\eqref{eq:ready2hydro} and write everything in terms of moments of $I^N_t$. We have 
\begin{align*}
\frac{d}{dt}m^N(t)&=  \sum_{k=0}^N k [a_{k-1}p_{k-1}(t) - (a_k + c_k)p_k(t) + c_{k+1}p_{k+1}(t)] + \beta \sum_{\ell = 2}^{N-1} (N -\ell) {\ell \choose 2} p_{\ell}(t)\\
&= \tau N \E(I^N_t)- \gamma \E(I^N_t) - \tau \E((I^N_t)^2) + \beta \Big( (N -I^N_t){I^N_t \choose 2} \Big)\\
&= \tau N \E(I^N_t)- \gamma \E(I^N_t) - \tau \E((I^N_t)^2) + \frac{\beta}{2} \big( N\E(I^N_t(I^N_t -1)) - \E((I^N_t-1)(I^N_t)^2)) \big)\\
&= \Big( \tau N -\gamma - \frac{N \beta}{2}\Big) \E(I^N_t) + \Big( -\tau + \frac{\beta(N+1)}{2} \Big) \E((I^N_t)^2) - \frac{\beta}{2}  \E((I^N_t)^3). 
\end{align*}
Now we scale everything by $N$ to obtain 
\[
\frac{d}{dt}\frac{m^N(t)}{N} = \Big( \tau N -\gamma - \frac{N \beta}{2}\Big) \E\Big(\frac{I^N_t}{N}\Big) + \Big( -\tau + \frac{\beta(N+1)}{2} \Big) N \E\Big(\Big (\frac{I^N_t}{N}\Big)^2\Big) - \frac{\beta}{2}  N^2 \E\Big(\Big (\frac{I^N_t}{N}\Big)^3\Big) . 
\]
At this point we substitute in the values $\tau = \lambda N^{-1},\beta = \mu N^{-2}$ to obtain 
\begin{align*}
\frac{d}{dt}\frac{m^N(t)}{N}&= \Big( \lambda -\gamma - \frac{\mu}{2N}\Big) \E\Big(\frac{I^N_t}{N}\Big) + \Big( -\lambda + \frac{\mu(N+1)}{2N} \Big) \E\Big(\Big (\frac{I^N_t}{N}\Big)^2\Big) - \frac{\mu}{2}\E\Big(\Big (\frac{I^N_t}{N}\Big)^3\Big)\\
&=  \Big( \lambda -\gamma \Big) \E\Big(\frac{I^N_t}{N}\Big) - \Big( \lambda - \frac{\mu}{2} \Big) \E\Big(\Big (\frac{I^N_t}{N}\Big)^2\Big) - \frac{\mu}{2}\E\Big(\Big (\frac{I^N_t}{N}\Big)^3\Big) - \frac{\mu}{2N}( \E\Big(\frac{I^N_t}{N}\Big)- \E\Big(\Big (\frac{I^N_t}{N}\Big)^2\Big)\Big)\\
&=  \Big( \lambda -\gamma \Big) \E\Big(\frac{I^N_t}{N}\Big) - \Big(\lambda - \frac{\mu}{2} \Big) \E\Big(\Big (\frac{I^N_t}{N}\Big)^2\Big) - \frac{\mu}{2}\E\Big(\Big (\frac{I^N_t}{N}\Big)^3\Big) + \mathcal O\Big( \frac{\mu}{N}\Big).
\end{align*}

Going forward, we will be keeping track of the error (which vanishes as $N \to \infty$) to make sure it will be independent of the time parameter $t$. We will need this fact at the end of the proof to interchange two limits. Above we used that $0 \le I^N_t \le 1$.  At this point we can substitute in \eqref{eq:rep1} and obtain 
\begin{align*}
\frac{d}{dt}\frac{m^N(t)}{N}&=\Big( \lambda -\gamma \Big) \E\Big(\frac{m^N(t)}{N} + \eta^N_t\Big) + \Big( -\lambda + \frac{\mu}{2} \Big) \E\Big(\Big (\frac{m^N(t)}{N} + \eta^N_t\Big)^2\Big) \\
&\phantom{xxxxxxxxxxxxxxxxxx}- \frac{\mu}{2}\E\Big(\Big (\frac{m^N(t)}{N} + \eta^N_t\Big)^3\Big) + \mathcal O\Big( \frac{\mu}{N}\Big)\\
&=\Big( \lambda -\gamma \Big) \Big(\frac{m^N(t)}{N} + \E\eta^N_t\Big) + \Big( -\lambda + \frac{\mu}{2} \Big)\Big( \Big(\frac{m^N(t)}{N}\Big)^2 + 2\frac{m^N(t)}{N} \E \eta^N_t + \E((\eta^N_t)^2)\Big) \\
&\phantom{xxxxxxx} - \frac{\mu}{2}\Big( \Big(\frac{m^N(t)}{N}\Big)^3 + 3\Big(\frac{m^N(t)}{N}\Big)^2 \E \eta^N_t++ 3\frac{m^N(t)}{N} \E(( \eta^N_t)^2) + \E((\eta^N_t)^3)\Big)  + \mathcal O\Big( \frac{\mu}{N}\Big).
\end{align*}

Since the $m^N(t)N^{-1}$ converge uniformly for $t \in [0,T]$, the error
\[
\varepsilon_{N,T} = \max_{p=1,2,3} \sup_{t \in [0,T]}\Big\| \Big(\frac{m^N(t)}{N}\Big)^p - (m_I(t))^p\Big\|_\infty
\]
converges to 0 as $N \to \infty$. We use this to write
\begin{align*}
\frac{d}{dt}\frac{m^N(t)}{N}&= ( \lambda -\gamma ) (m_I(t)+ \E\eta^N_t ) + \Big( -\lambda + \frac{\mu}{2} \Big)( m_I^2(t) + 2 m_I^2(t) \E \eta^N_t + \E((\eta^N_t)^2)) \\
&\phantom{xxxxxxx} - \frac{\mu}{2}( m_I^3(t) + 3 m_I^2(t)  \E \eta^N_t+ 3 m_I(t) \E(( \eta^N_t)^2) + \E((\eta^N_t)^3))  + \mathcal O\Big( \frac{\mu}{N} \vee \varepsilon_{N,T}\Big).
\end{align*}
For $N$ large enough, $\eta^N_t \le 1$. This is because $m^N(t) N^{-1} \le 1$ and it will be uniformly close to its limit. Therefore, condition \eqref{eq:ass0} implies 
\[
\delta_{N, T} = \sup_{t \le T} \E(\eta^N_t) \ge  \sup_{t \le T} \E((\eta^N_t)^p), \quad \text{ and } \delta_{N, T} \to 0. 
\]
Therefore, with a final substitution we have reached that for $N$ large enough
 \be\label{eq:finale}
 \frac{d}{dt}\frac{m^N(t)}{N}= ( \lambda -\gamma )m_I(t) + \Big( -\lambda + \frac{\mu}{2} \Big)m_I^2(t) - \frac{\mu}{2} m_I^3(t) + \mathcal O\Big( \frac{\mu}{N} \vee \varepsilon_{N,T}\vee \delta_{N, T}\Big).
 \ee
Note that as $N$ grows, the error term in \eqref{eq:finale} will vanish in the limit. 
Therefore, we want to argue that 
\be \label{eq:unif-con}
\lim_{N \to \infty} \frac{d}{dt} \frac{ m^N(t) }{N} = \frac{d}{dt}  \lim_{N \to \infty} \frac{ m^N(t) }{N} =  \frac{d}{dt}m_I(t).
\ee
From the assumptions we have the pointwise convergence of  $\frac{ m^N(t) }{N} $. From Eq.~\eqref{eq:finale} we have that the derivatives of $ \frac{m^N(t)}{N}$ converge uniformly in any interval $[0,T]$. Then, from \cite{rudin1953principles} Theorem 7.17 we can exchange limits and \eqref{eq:unif-con} is valid. 

The proof of Eq.~\eqref{eq:ODE1} now concludes by taking the limit as $N \to \infty$ in \eqref{eq:finale} to obtain 
\be \label{2simhydro}
 \frac{d}{dt}m_I(t) = ( \lambda -\gamma )m_I(t) + \Big(\frac{\mu}{2} - \lambda \Big)m_I^2(t) - 
\frac{\mu}{2} m_I^3(t).
\ee

The proof of \eqref{eq:lim0} follows from \eqref{eq:rep1}, \eqref{eq:unipo} and the weak convergence of $\{\eta_t\}_{t\le T}$ to 0, guaranteed by \eqref{eq:ass0}. 
\end{proof}

\subsection{Mean-field limit for complete simplicial 3-complex}
\begin{proof}[Proof of Theorem~\ref{thm:3-simplexHydro}]
    The ODE for the un-scaled first moment is 
\begin{align*} 
\frac{d}{dt} m^N(t) &= \sum_{k=0}^N k \frac{d}{dt} p_k(t)\\
&= \sum_{k=0}^N k \left( a_{k-1}p_{k-1}(t) +c_{k+1}p_{k+1}(t) - (a_k + c_k)p_k \right)
+ \beta \sum_{\ell=2}^{N-1}(N-\ell){\ell \choose 2}p_{\ell}(t)\\
& \phantom{xxxxxxxxxxxxx} + \sum_{k=0}^N k \left( \delta^{\square}_{k-1}p_{k-1}(t) - \delta_k^{\square} p_k(t)\right).
\end{align*}
Note that the first line of that equality is actually treated and scaled by the argument in Theorem \ref{thm:hydroKn}, so we can control under the same assumptions its limit. 
So we focus on the sum on the second line. We compute, taking into accounts the boundaries,  
\begin{align*}
  \sum_{k=0}^N k \left( \delta^{\square}_{k-1}p_{k-1}(t) - \delta_k^{\square} p_k(t)\right) & = \sum_{k=4}^N k  \delta^{\square}_{k-1}p_{k-1}(t) - \sum_{k=3}^{N-1}k \delta_k^{\square} p_k(t)\\
  &= \sum_{k=3}^{N-1} (k+1)  \delta^{\square}_{k}p_{k}(t) - \sum_{k=3}^{N-1}k \delta_k^{\square} p_k(t)\\
  &= \sum_{k=3}^{N-1} \delta^{\square}_{k}p_{k}(t) = \sum_{k=3}^{N-1} \delta {k \choose 3}(N-k)p_{k}(t)\\
  &= \delta \E \left( \frac{I^N_t(I^N_t -1)(I^N_t-2)}{6} (N- I^N_t)\right).
\end{align*}

In order to find the appropriate scaling for $\delta$ we need to expand the expression in the expectation and it yields 
\[
 \delta \E \left( \frac{I^N_t(I^N_t -1)(I^N_t-2)}{6} (N- I^N_t)\right) = \frac{\delta}{6}\E(-(I^N_t)^4+(N+3)(I^N_t)^3+(-3N-2)(I^N_t)^2+2NI^N_t)
\]
So we will scale $\delta = \theta N^{-3}$. Then we have 
\[
\frac{\theta}{6}\E\left(-N \frac{(I^N_t)^4}{N^4}+(N+3)\frac{(I^N_t)^3}{N^3}+\frac{(-3N-2)}{N}\frac{(I^N_t)^2}{N^2}+2\frac{I^N_t}{N^2}\right)
\]

To put everything together, recall that as before, 
\[ \tau = \lambda N^{-1}, \quad \beta = \mu N^{-2}, \quad \delta = \theta N^{-3} \quad I^N_t N^{-1} \sim m_I(t) + \eta_N(t), \quad \eta_N(t) \to 0 \]
and divide the differential equation for $m_N(t)$ by $N$ and take a limit as $N\to \infty$ to obtain
\begin{align}  \label{ode3simplex}
\frac{d}{dt}m_I(t) &= \lim_{N\to \infty} \frac{d}{dt}\frac{m_N(t)}{N} = (\lambda - \gamma) m_I(t) + \Big( \frac{\mu}{2} - \lambda\Big) m_I(t)^2 -\frac{\mu}{2} m_I(t)^3 \notag \\
&\phantom{xxxxx}+ \lim_{N\to \infty} 
\frac{1}{N}\frac{\theta}{6}\E\left(-N \frac{(I^N_t)^4}{N^4}+(N+3)\frac{(I^N_t)^3}{N^3}+\frac{(-3N-2)}{N}\frac{(I^N_t)^2}{N^2}+2\frac{I^N_t}{N^2}\right) \notag \\
&= (\lambda - \gamma) m_I(t) + \Big( \frac{\mu}{2} - \lambda\Big) m_I(t)^2 +\left(\frac{\theta}{6}-\frac{\mu}{2}\right) m_I(t)^3 - \frac{\theta}{6} (m_{I}(t))^4.
\end{align}

\end{proof}

\bibliography{project1.bib}

\begin{thebibliography}{49}%
\makeatletter
\providecommand \@ifxundefined [1]{%
 \@ifx{#1\undefined}
}%
\providecommand \@ifnum [1]{%
 \ifnum #1\expandafter \@firstoftwo
 \else \expandafter \@secondoftwo
 \fi
}%
\providecommand \@ifx [1]{%
 \ifx #1\expandafter \@firstoftwo
 \else \expandafter \@secondoftwo
 \fi
}%
\providecommand \natexlab [1]{#1}%
\providecommand \enquote  [1]{``#1''}%
\providecommand \bibnamefont  [1]{#1}%
\providecommand \bibfnamefont [1]{#1}%
\providecommand \citenamefont [1]{#1}%
\providecommand \href@noop [0]{\@secondoftwo}%
\providecommand \href [0]{\begingroup \@sanitize@url \@href}%
\providecommand \@href[1]{\@@startlink{#1}\@@href}%
\providecommand \@@href[1]{\endgroup#1\@@endlink}%
\providecommand \@sanitize@url [0]{\catcode `\\12\catcode `\$12\catcode
  `\&12\catcode `\#12\catcode `\^12\catcode `\_12\catcode `\%12\relax}%
\providecommand \@@startlink[1]{}%
\providecommand \@@endlink[0]{}%
\providecommand \url  [0]{\begingroup\@sanitize@url \@url }%
\providecommand \@url [1]{\endgroup\@href {#1}{\urlprefix }}%
\providecommand \urlprefix  [0]{URL }%
\providecommand \Eprint [0]{\href }%
\providecommand \doibase [0]{https://doi.org/}%
\providecommand \selectlanguage [0]{\@gobble}%
\providecommand \bibinfo  [0]{\@secondoftwo}%
\providecommand \bibfield  [0]{\@secondoftwo}%
\providecommand \translation [1]{[#1]}%
\providecommand \BibitemOpen [0]{}%
\providecommand \bibitemStop [0]{}%
\providecommand \bibitemNoStop [0]{.\EOS\space}%
\providecommand \EOS [0]{\spacefactor3000\relax}%
\providecommand \BibitemShut  [1]{\csname bibitem#1\endcsname}%
\let\auto@bib@innerbib\@empty
\bibitem [{\citenamefont {Albert}\ and\ \citenamefont
  {Barab{\'a}si}(2002)}]{albert2002statistical}%
  \BibitemOpen
  \bibfield  {author} {\bibinfo {author} {\bibfnamefont {R.}~\bibnamefont
  {Albert}}\ and\ \bibinfo {author} {\bibfnamefont {A.-L.}\ \bibnamefont
  {Barab{\'a}si}},\ }\bibfield  {title} {\bibinfo {title} {Statistical
  mechanics of complex networks},\ }\href
  {https://doi.org/https://doi.org/10.1103/RevModPhys.74.47} {\bibfield
  {journal} {\bibinfo  {journal} {Rev. Mod. Phys.}\ }\textbf {\bibinfo {volume}
  {74}},\ \bibinfo {pages} {47} (\bibinfo {year} {2002})}\BibitemShut {NoStop}%
\bibitem [{\citenamefont {Newman}(2003)}]{newman2003structure}%
  \BibitemOpen
  \bibfield  {author} {\bibinfo {author} {\bibfnamefont {M.~E.}\ \bibnamefont
  {Newman}},\ }\bibfield  {title} {\bibinfo {title} {The structure and function
  of complex networks},\ }\href
  {https://doi.org/https://doi.org/10.1137/S003614450342480} {\bibfield
  {journal} {\bibinfo  {journal} {SIAM review}\ }\textbf {\bibinfo {volume}
  {45}},\ \bibinfo {pages} {167} (\bibinfo {year} {2003})}\BibitemShut
  {NoStop}%
\bibitem [{\citenamefont {Barrat}\ \emph {et~al.}(2008)\citenamefont {Barrat},
  \citenamefont {Barth{\'e}lemy},\ and\ \citenamefont
  {Vespignani}}]{barrat2008dynamical}%
  \BibitemOpen
  \bibfield  {author} {\bibinfo {author} {\bibfnamefont {A.}~\bibnamefont
  {Barrat}}, \bibinfo {author} {\bibfnamefont {M.}~\bibnamefont
  {Barth{\'e}lemy}},\ and\ \bibinfo {author} {\bibfnamefont {A.}~\bibnamefont
  {Vespignani}},\ }\href {https://books.google.at/books?id=TmgePn9uQD4C} {\emph
  {\bibinfo {title} {Dynamical Processes on Complex Networks}}}\ (\bibinfo
  {publisher} {Cambridge University Press},\ \bibinfo {year}
  {2008})\BibitemShut {NoStop}%
\bibitem [{\citenamefont {Latora}\ \emph {et~al.}(2017)\citenamefont {Latora},
  \citenamefont {Nicosia},\ and\ \citenamefont
  {Russo}}]{latora_nicosia_russo_2017}%
  \BibitemOpen
  \bibfield  {author} {\bibinfo {author} {\bibfnamefont {V.}~\bibnamefont
  {Latora}}, \bibinfo {author} {\bibfnamefont {V.}~\bibnamefont {Nicosia}},\
  and\ \bibinfo {author} {\bibfnamefont {G.}~\bibnamefont {Russo}},\ }\href
  {https://books.google.it/books?id=qV0yDwAAQBAJ} {\emph {\bibinfo {title}
  {Complex Networks: Principles, Methods and Applications}}},\ Complex
  Networks: Principles, Methods and Applications\ (\bibinfo  {publisher}
  {Cambridge University Press},\ \bibinfo {year} {2017})\BibitemShut {NoStop}%
\bibitem [{\citenamefont {Lambiotte}\ \emph {et~al.}(2019)\citenamefont
  {Lambiotte}, \citenamefont {Rosvall},\ and\ \citenamefont
  {Scholtes}}]{lambiotte2019networks}%
  \BibitemOpen
  \bibfield  {author} {\bibinfo {author} {\bibfnamefont {R.}~\bibnamefont
  {Lambiotte}}, \bibinfo {author} {\bibfnamefont {M.}~\bibnamefont {Rosvall}},\
  and\ \bibinfo {author} {\bibfnamefont {I.}~\bibnamefont {Scholtes}},\
  }\bibfield  {title} {\bibinfo {title} {From networks to optimal higher-order
  models of complex systems},\ }\bibfield  {journal} {\bibinfo  {journal} {Nat.
  Phys.}\ }\href {https://doi.org/10.1038/s41567-019-0459-y}
  {10.1038/s41567-019-0459-y} (\bibinfo {year} {2019})\BibitemShut {NoStop}%
\bibitem [{\citenamefont {Battiston}\ \emph {et~al.}(2021)\citenamefont
  {Battiston}, \citenamefont {Amico}, \citenamefont {Barrat}, \citenamefont
  {Bianconi}, \citenamefont {Ferraz~de Arruda}, \citenamefont {Franceschiello},
  \citenamefont {Iacopini}, \citenamefont {K{\'e}fi}, \citenamefont {Latora},
  \citenamefont {Moreno}, \citenamefont {Murray}, \citenamefont {Peixoto},
  \citenamefont {Vaccarino},\ and\ \citenamefont
  {Petri}}]{battiston2021physics}%
  \BibitemOpen
  \bibfield  {author} {\bibinfo {author} {\bibfnamefont {F.}~\bibnamefont
  {Battiston}}, \bibinfo {author} {\bibfnamefont {E.}~\bibnamefont {Amico}},
  \bibinfo {author} {\bibfnamefont {A.}~\bibnamefont {Barrat}}, \bibinfo
  {author} {\bibfnamefont {G.}~\bibnamefont {Bianconi}}, \bibinfo {author}
  {\bibfnamefont {G.}~\bibnamefont {Ferraz~de Arruda}}, \bibinfo {author}
  {\bibfnamefont {B.}~\bibnamefont {Franceschiello}}, \bibinfo {author}
  {\bibfnamefont {I.}~\bibnamefont {Iacopini}}, \bibinfo {author}
  {\bibfnamefont {S.}~\bibnamefont {K{\'e}fi}}, \bibinfo {author}
  {\bibfnamefont {V.}~\bibnamefont {Latora}}, \bibinfo {author} {\bibfnamefont
  {Y.}~\bibnamefont {Moreno}}, \bibinfo {author} {\bibfnamefont
  {M.}~\bibnamefont {Murray}}, \bibinfo {author} {\bibfnamefont
  {T.}~\bibnamefont {Peixoto}}, \bibinfo {author} {\bibfnamefont
  {F.}~\bibnamefont {Vaccarino}},\ and\ \bibinfo {author} {\bibfnamefont
  {G.}~\bibnamefont {Petri}},\ }\bibfield  {title} {\bibinfo {title} {The
  physics of higher-order interactions in complex systems},\ }\href
  {https://doi.org/10.1038/s41567-021-01371-4} {\bibfield  {journal} {\bibinfo
  {journal} {Nat. Phys.}\ }\textbf {\bibinfo {volume} {17}},\ \bibinfo {pages}
  {1093} (\bibinfo {year} {2021})}\BibitemShut {NoStop}%
\bibitem [{\citenamefont {Bick}\ \emph {et~al.}(2023)\citenamefont {Bick},
  \citenamefont {Gross}, \citenamefont {Harrington},\ and\ \citenamefont
  {Schaub}}]{bick2023higher}%
  \BibitemOpen
  \bibfield  {author} {\bibinfo {author} {\bibfnamefont {C.}~\bibnamefont
  {Bick}}, \bibinfo {author} {\bibfnamefont {E.}~\bibnamefont {Gross}},
  \bibinfo {author} {\bibfnamefont {H.~A.}\ \bibnamefont {Harrington}},\ and\
  \bibinfo {author} {\bibfnamefont {M.~T.}\ \bibnamefont {Schaub}},\ }\bibfield
   {title} {\bibinfo {title} {What are higher-order networks?},\ }\href
  {https://doi.org/https://doi.org/10.1137/21M1414024} {\bibfield  {journal}
  {\bibinfo  {journal} {SIAM Rev.}\ }\textbf {\bibinfo {volume} {65}},\
  \bibinfo {pages} {686} (\bibinfo {year} {2023})}\BibitemShut {NoStop}%
\bibitem [{\citenamefont {Wasserman}\ and\ \citenamefont
  {Faust}(1994)}]{wasserman1994social}%
  \BibitemOpen
  \bibfield  {author} {\bibinfo {author} {\bibfnamefont {S.}~\bibnamefont
  {Wasserman}}\ and\ \bibinfo {author} {\bibfnamefont {K.}~\bibnamefont
  {Faust}},\ }\href {https://books.google.at/books?id=CAm2DpIqRUIC} {\emph
  {\bibinfo {title} {Social Network Analysis: Methods and Applications}}},\
  Structural Analysis in the Social Sciences\ (\bibinfo  {publisher} {Cambridge
  University Press},\ \bibinfo {year} {1994})\BibitemShut {NoStop}%
\bibitem [{\citenamefont {Patania}\ \emph {et~al.}(2017)\citenamefont
  {Patania}, \citenamefont {Petri},\ and\ \citenamefont
  {Vaccarino}}]{patania2017shape}%
  \BibitemOpen
  \bibfield  {author} {\bibinfo {author} {\bibfnamefont {A.}~\bibnamefont
  {Patania}}, \bibinfo {author} {\bibfnamefont {G.}~\bibnamefont {Petri}},\
  and\ \bibinfo {author} {\bibfnamefont {F.}~\bibnamefont {Vaccarino}},\
  }\bibfield  {title} {\bibinfo {title} {The shape of collaborations},\ }\href
  {https://doi.org/https://doi.org/10.1140/epjds/s13688-017-0114-8} {\bibfield
  {journal} {\bibinfo  {journal} {EPJ Data Sci.}\ }\textbf {\bibinfo {volume}
  {6}},\ \bibinfo {pages} {1} (\bibinfo {year} {2017})}\BibitemShut {NoStop}%
\bibitem [{\citenamefont {Benson}\ \emph {et~al.}(2018)\citenamefont {Benson},
  \citenamefont {Abebe}, \citenamefont {Schaub}, \citenamefont {Jadbabaie},\
  and\ \citenamefont {Kleinberg}}]{benson2018simplicial}%
  \BibitemOpen
  \bibfield  {author} {\bibinfo {author} {\bibfnamefont {A.~R.}\ \bibnamefont
  {Benson}}, \bibinfo {author} {\bibfnamefont {R.}~\bibnamefont {Abebe}},
  \bibinfo {author} {\bibfnamefont {M.~T.}\ \bibnamefont {Schaub}}, \bibinfo
  {author} {\bibfnamefont {A.}~\bibnamefont {Jadbabaie}},\ and\ \bibinfo
  {author} {\bibfnamefont {J.}~\bibnamefont {Kleinberg}},\ }\bibfield  {title}
  {\bibinfo {title} {Simplicial closure and higher-order link prediction},\
  }\href {https://doi.org/https://doi.org/10.1073/pnas.1800683115} {\bibfield
  {journal} {\bibinfo  {journal} {Proc. Natl. Acad. Sci. U.S.A}\ }\textbf
  {\bibinfo {volume} {115}},\ \bibinfo {pages} {E11221} (\bibinfo {year}
  {2018})}\BibitemShut {NoStop}%
\bibitem [{\citenamefont {Iacopini}\ \emph
  {et~al.}(2023{\natexlab{a}})\citenamefont {Iacopini}, \citenamefont {Foote},
  \citenamefont {Fefferman}, \citenamefont {Derryberry},\ and\ \citenamefont
  {Silk}}]{iacopini2023not}%
  \BibitemOpen
  \bibfield  {author} {\bibinfo {author} {\bibfnamefont {I.}~\bibnamefont
  {Iacopini}}, \bibinfo {author} {\bibfnamefont {J.~R.}\ \bibnamefont {Foote}},
  \bibinfo {author} {\bibfnamefont {N.~H.}\ \bibnamefont {Fefferman}}, \bibinfo
  {author} {\bibfnamefont {E.~P.}\ \bibnamefont {Derryberry}},\ and\ \bibinfo
  {author} {\bibfnamefont {M.~J.}\ \bibnamefont {Silk}},\ }\bibfield  {title}
  {\bibinfo {title} {Not your private t\^ete-\`a-t\^ete: leveraging the power
  of higher-order networks to study animal communication},\ }\bibfield
  {journal} {\bibinfo  {journal} {arXiv preprint arXiv:2309.03783}\ }\href
  {https://doi.org/https://doi.org/10.48550/arXiv.2309.03783}
  {https://doi.org/10.48550/arXiv.2309.03783} (\bibinfo {year}
  {2023}{\natexlab{a}})\BibitemShut {NoStop}%
\bibitem [{\citenamefont {Battiston}\ \emph {et~al.}(2020)\citenamefont
  {Battiston}, \citenamefont {Cencetti}, \citenamefont {Iacopini},
  \citenamefont {Latora}, \citenamefont {Lucas}, \citenamefont {Patania},
  \citenamefont {Young},\ and\ \citenamefont {Petri}}]{battiston2020networks}%
  \BibitemOpen
  \bibfield  {author} {\bibinfo {author} {\bibfnamefont {F.}~\bibnamefont
  {Battiston}}, \bibinfo {author} {\bibfnamefont {G.}~\bibnamefont {Cencetti}},
  \bibinfo {author} {\bibfnamefont {I.}~\bibnamefont {Iacopini}}, \bibinfo
  {author} {\bibfnamefont {V.}~\bibnamefont {Latora}}, \bibinfo {author}
  {\bibfnamefont {M.}~\bibnamefont {Lucas}}, \bibinfo {author} {\bibfnamefont
  {A.}~\bibnamefont {Patania}}, \bibinfo {author} {\bibfnamefont {J.-G.}\
  \bibnamefont {Young}},\ and\ \bibinfo {author} {\bibfnamefont
  {G.}~\bibnamefont {Petri}},\ }\bibfield  {title} {\bibinfo {title} {Networks
  beyond pairwise interactions: {{Structure}} and dynamics},\ }\href
  {https://doi.org/10.1016/j.physrep.2020.05.004} {\bibfield  {journal}
  {\bibinfo  {journal} {Phys. Rep.}\ }\textbf {\bibinfo {volume} {874}},\
  \bibinfo {pages} {1} (\bibinfo {year} {2020})}\BibitemShut {NoStop}%
\bibitem [{\citenamefont {Torres}\ \emph {et~al.}(2021)\citenamefont {Torres},
  \citenamefont {Blevins}, \citenamefont {Bassett},\ and\ \citenamefont
  {Eliassi-Rad}}]{torres2021and}%
  \BibitemOpen
  \bibfield  {author} {\bibinfo {author} {\bibfnamefont {L.}~\bibnamefont
  {Torres}}, \bibinfo {author} {\bibfnamefont {A.~S.}\ \bibnamefont {Blevins}},
  \bibinfo {author} {\bibfnamefont {D.}~\bibnamefont {Bassett}},\ and\ \bibinfo
  {author} {\bibfnamefont {T.}~\bibnamefont {Eliassi-Rad}},\ }\bibfield
  {title} {\bibinfo {title} {The why, how, and when of representations for
  complex systems},\ }\href
  {https://doi.org/https://doi.org/10.1137/20M1355896} {\bibfield  {journal}
  {\bibinfo  {journal} {SIAM Rev.}\ }\textbf {\bibinfo {volume} {63}},\
  \bibinfo {pages} {435} (\bibinfo {year} {2021})}\BibitemShut {NoStop}%
\bibitem [{\citenamefont {Skardal}\ and\ \citenamefont
  {Arenas}(2019)}]{skardal2019abrupt}%
  \BibitemOpen
  \bibfield  {author} {\bibinfo {author} {\bibfnamefont {P.~S.}\ \bibnamefont
  {Skardal}}\ and\ \bibinfo {author} {\bibfnamefont {A.}~\bibnamefont
  {Arenas}},\ }\bibfield  {title} {\bibinfo {title} {Abrupt desynchronization
  and extensive multistability in globally coupled oscillator simplexes},\
  }\href {https://doi.org/10.1103/PhysRevLett.122.248301} {\bibfield  {journal}
  {\bibinfo  {journal} {Phys. Rev. Lett.}\ }\textbf {\bibinfo {volume} {122}},\
  \bibinfo {pages} {248301} (\bibinfo {year} {2019})}\BibitemShut {NoStop}%
\bibitem [{\citenamefont {Mill{\'a}n}\ \emph {et~al.}(2020)\citenamefont
  {Mill{\'a}n}, \citenamefont {Torres},\ and\ \citenamefont
  {Bianconi}}]{millan2020explosive}%
  \BibitemOpen
  \bibfield  {author} {\bibinfo {author} {\bibfnamefont {A.~P.}\ \bibnamefont
  {Mill{\'a}n}}, \bibinfo {author} {\bibfnamefont {J.~J.}\ \bibnamefont
  {Torres}},\ and\ \bibinfo {author} {\bibfnamefont {G.}~\bibnamefont
  {Bianconi}},\ }\bibfield  {title} {\bibinfo {title} {Explosive {{Higher-Order
  Kuramoto Dynamics}} on {{Simplicial Complexes}}},\ }\href
  {https://doi.org/10.1103/PhysRevLett.124.218301} {\bibfield  {journal}
  {\bibinfo  {journal} {Phys. Rev. Lett.}\ }\textbf {\bibinfo {volume} {124}},\
  \bibinfo {pages} {218301} (\bibinfo {year} {2020})}\BibitemShut {NoStop}%
\bibitem [{\citenamefont {Neuh{\"a}user}\ \emph {et~al.}(2020)\citenamefont
  {Neuh{\"a}user}, \citenamefont {Mellor},\ and\ \citenamefont
  {Lambiotte}}]{neuhauser2020multibody}%
  \BibitemOpen
  \bibfield  {author} {\bibinfo {author} {\bibfnamefont {L.}~\bibnamefont
  {Neuh{\"a}user}}, \bibinfo {author} {\bibfnamefont {A.}~\bibnamefont
  {Mellor}},\ and\ \bibinfo {author} {\bibfnamefont {R.}~\bibnamefont
  {Lambiotte}},\ }\bibfield  {title} {\bibinfo {title} {Multibody interactions
  and nonlinear consensus dynamics on networked systems},\ }\href
  {https://doi.org/10.1103/PhysRevE.101.032310} {\bibfield  {journal} {\bibinfo
   {journal} {Phys. Rev. E}\ }\textbf {\bibinfo {volume} {101}},\ \bibinfo
  {pages} {032310} (\bibinfo {year} {2020})}\BibitemShut {NoStop}%
\bibitem [{\citenamefont {Alvarez-Rodriguez}\ \emph {et~al.}(2021)\citenamefont
  {Alvarez-Rodriguez}, \citenamefont {Battiston}, \citenamefont {de~Arruda},
  \citenamefont {Moreno}, \citenamefont {Perc},\ and\ \citenamefont
  {Latora}}]{alvarez2021evolutionary}%
  \BibitemOpen
  \bibfield  {author} {\bibinfo {author} {\bibfnamefont {U.}~\bibnamefont
  {Alvarez-Rodriguez}}, \bibinfo {author} {\bibfnamefont {F.}~\bibnamefont
  {Battiston}}, \bibinfo {author} {\bibfnamefont {G.~F.}\ \bibnamefont
  {de~Arruda}}, \bibinfo {author} {\bibfnamefont {Y.}~\bibnamefont {Moreno}},
  \bibinfo {author} {\bibfnamefont {M.}~\bibnamefont {Perc}},\ and\ \bibinfo
  {author} {\bibfnamefont {V.}~\bibnamefont {Latora}},\ }\bibfield  {title}
  {\bibinfo {title} {Evolutionary dynamics of higher-order interactions in
  social networks},\ }\href
  {https://doi.org/https://doi.org/10.1038/s41562-020-01024-1} {\bibfield
  {journal} {\bibinfo  {journal} {Nat. Hum. Behav.}\ }\textbf {\bibinfo
  {volume} {5}},\ \bibinfo {pages} {586} (\bibinfo {year} {2021})}\BibitemShut
  {NoStop}%
\bibitem [{\citenamefont {Schawe}\ and\ \citenamefont
  {Hern{\'a}ndez}(2022)}]{schawe2022higher}%
  \BibitemOpen
  \bibfield  {author} {\bibinfo {author} {\bibfnamefont {H.}~\bibnamefont
  {Schawe}}\ and\ \bibinfo {author} {\bibfnamefont {L.}~\bibnamefont
  {Hern{\'a}ndez}},\ }\bibfield  {title} {\bibinfo {title} {Higher order
  interactions destroy phase transitions in deffuant opinion dynamics model},\
  }\href {https://doi.org/https://doi.org/10.1038/s42005-022-00807-4}
  {\bibfield  {journal} {\bibinfo  {journal} {Commun. Phys.}\ }\textbf
  {\bibinfo {volume} {5}},\ \bibinfo {pages} {32} (\bibinfo {year}
  {2022})}\BibitemShut {NoStop}%
\bibitem [{\citenamefont {Barrat}\ \emph {et~al.}(2022)\citenamefont {Barrat},
  \citenamefont {Ferraz~de Arruda}, \citenamefont {Iacopini},\ and\
  \citenamefont {Moreno}}]{barrat2022social}%
  \BibitemOpen
  \bibfield  {author} {\bibinfo {author} {\bibfnamefont {A.}~\bibnamefont
  {Barrat}}, \bibinfo {author} {\bibfnamefont {G.}~\bibnamefont {Ferraz~de
  Arruda}}, \bibinfo {author} {\bibfnamefont {I.}~\bibnamefont {Iacopini}},\
  and\ \bibinfo {author} {\bibfnamefont {Y.}~\bibnamefont {Moreno}},\
  }\bibfield  {title} {\bibinfo {title} {Social contagion on higher-order
  structures},\ }in\ \href
  {https://doi.org/https://doi.org/10.1007/978-3-030-91374-8_13} {\emph
  {\bibinfo {booktitle} {Higher-Order Systems}}}\ (\bibinfo  {publisher}
  {Springer},\ \bibinfo {year} {2022})\ pp.\ \bibinfo {pages}
  {329--346}\BibitemShut {NoStop}%
\bibitem [{\citenamefont {Papanikolaou}\ \emph {et~al.}(2022)\citenamefont
  {Papanikolaou}, \citenamefont {Vaccario}, \citenamefont {Hormann},
  \citenamefont {Lambiotte},\ and\ \citenamefont
  {Schweitzer}}]{papanikolaou2022consensus}%
  \BibitemOpen
  \bibfield  {author} {\bibinfo {author} {\bibfnamefont {N.}~\bibnamefont
  {Papanikolaou}}, \bibinfo {author} {\bibfnamefont {G.}~\bibnamefont
  {Vaccario}}, \bibinfo {author} {\bibfnamefont {E.}~\bibnamefont {Hormann}},
  \bibinfo {author} {\bibfnamefont {R.}~\bibnamefont {Lambiotte}},\ and\
  \bibinfo {author} {\bibfnamefont {F.}~\bibnamefont {Schweitzer}},\ }\bibfield
   {title} {\bibinfo {title} {Consensus from group interactions: An adaptive
  voter model on hypergraphs},\ }\href
  {https://doi.org/10.1103/PhysRevE.105.054307} {\bibfield  {journal} {\bibinfo
   {journal} {Phys. Rev. E}\ }\textbf {\bibinfo {volume} {105}},\ \bibinfo
  {pages} {054307} (\bibinfo {year} {2022})}\BibitemShut {NoStop}%
\bibitem [{\citenamefont {Iacopini}\ \emph {et~al.}(2019)\citenamefont
  {Iacopini}, \citenamefont {Petri}, \citenamefont {Barrat},\ and\
  \citenamefont {Latora}}]{iacopini2019simplicial}%
  \BibitemOpen
  \bibfield  {author} {\bibinfo {author} {\bibfnamefont {I.}~\bibnamefont
  {Iacopini}}, \bibinfo {author} {\bibfnamefont {G.}~\bibnamefont {Petri}},
  \bibinfo {author} {\bibfnamefont {A.}~\bibnamefont {Barrat}},\ and\ \bibinfo
  {author} {\bibfnamefont {V.}~\bibnamefont {Latora}},\ }\bibfield  {title}
  {\bibinfo {title} {Simplicial models of social contagion},\ }\href
  {https://doi.org/https://doi.org/10.1038/s41467-019-10431-6} {\bibfield
  {journal} {\bibinfo  {journal} {Nat. Commun.}\ }\textbf {\bibinfo {volume}
  {10}},\ \bibinfo {pages} {2485} (\bibinfo {year} {2019})}\BibitemShut
  {NoStop}%
\bibitem [{\citenamefont {Matamalas}\ \emph {et~al.}(2020)\citenamefont
  {Matamalas}, \citenamefont {G{\'o}mez},\ and\ \citenamefont
  {Arenas}}]{matamalas2020abrupt}%
  \BibitemOpen
  \bibfield  {author} {\bibinfo {author} {\bibfnamefont {J.~T.}\ \bibnamefont
  {Matamalas}}, \bibinfo {author} {\bibfnamefont {S.}~\bibnamefont
  {G{\'o}mez}},\ and\ \bibinfo {author} {\bibfnamefont {A.}~\bibnamefont
  {Arenas}},\ }\bibfield  {title} {\bibinfo {title} {Abrupt phase transition of
  epidemic spreading in simplicial complexes},\ }\href
  {https://doi.org/10.1103/PhysRevResearch.2.012049} {\bibfield  {journal}
  {\bibinfo  {journal} {Phys. Rev. Research}\ }\textbf {\bibinfo {volume}
  {2}},\ \bibinfo {pages} {012049} (\bibinfo {year} {2020})}\BibitemShut
  {NoStop}%
\bibitem [{\citenamefont {{de Arruda}}\ \emph {et~al.}(2020)\citenamefont {{de
  Arruda}}, \citenamefont {Petri},\ and\ \citenamefont
  {Moreno}}]{de2020social}%
  \BibitemOpen
  \bibfield  {author} {\bibinfo {author} {\bibfnamefont {G.~F.}\ \bibnamefont
  {{de Arruda}}}, \bibinfo {author} {\bibfnamefont {G.}~\bibnamefont {Petri}},\
  and\ \bibinfo {author} {\bibfnamefont {Y.}~\bibnamefont {Moreno}},\
  }\bibfield  {title} {\bibinfo {title} {Social contagion models on
  hypergraphs},\ }\href@noop {} {\bibfield  {journal} {\bibinfo  {journal}
  {Phys Rev Res}\ }\textbf {\bibinfo {volume} {2}},\ \bibinfo {pages} {023032}
  (\bibinfo {year} {2020})}\BibitemShut {NoStop}%
\bibitem [{\citenamefont {Landry}\ and\ \citenamefont
  {Restrepo}(2020)}]{landry2020effect}%
  \BibitemOpen
  \bibfield  {author} {\bibinfo {author} {\bibfnamefont {N.~W.}\ \bibnamefont
  {Landry}}\ and\ \bibinfo {author} {\bibfnamefont {J.~G.}\ \bibnamefont
  {Restrepo}},\ }\bibfield  {title} {\bibinfo {title} {The effect of
  heterogeneity on hypergraph contagion models},\ }\href
  {https://doi.org/10.1063/5.0020034} {\bibfield  {journal} {\bibinfo
  {journal} {Chaos}\ }\textbf {\bibinfo {volume} {30}},\ \bibinfo {pages}
  {103117} (\bibinfo {year} {2020})}\BibitemShut {NoStop}%
\bibitem [{\citenamefont {St-Onge}\ \emph {et~al.}(2022)\citenamefont
  {St-Onge}, \citenamefont {Iacopini}, \citenamefont {Latora}, \citenamefont
  {Barrat}, \citenamefont {Petri}, \citenamefont {Allard},\ and\ \citenamefont
  {H{\'e}bert-Dufresne}}]{st2022influential}%
  \BibitemOpen
  \bibfield  {author} {\bibinfo {author} {\bibfnamefont {G.}~\bibnamefont
  {St-Onge}}, \bibinfo {author} {\bibfnamefont {I.}~\bibnamefont {Iacopini}},
  \bibinfo {author} {\bibfnamefont {V.}~\bibnamefont {Latora}}, \bibinfo
  {author} {\bibfnamefont {A.}~\bibnamefont {Barrat}}, \bibinfo {author}
  {\bibfnamefont {G.}~\bibnamefont {Petri}}, \bibinfo {author} {\bibfnamefont
  {A.}~\bibnamefont {Allard}},\ and\ \bibinfo {author} {\bibfnamefont
  {L.}~\bibnamefont {H{\'e}bert-Dufresne}},\ }\bibfield  {title} {\bibinfo
  {title} {Influential groups for seeding and sustaining nonlinear contagion in
  heterogeneous hypergraphs},\ }\href
  {https://doi.org/doi.org/10.1038/s42005-021-00788-w} {\bibfield  {journal}
  {\bibinfo  {journal} {Commun. Phys.}\ }\textbf {\bibinfo {volume} {5}},\
  \bibinfo {pages} {1} (\bibinfo {year} {2022})}\BibitemShut {NoStop}%
\bibitem [{\citenamefont {Ferraz~de Arruda}\ \emph {et~al.}(2023)\citenamefont
  {Ferraz~de Arruda}, \citenamefont {Petri}, \citenamefont {Rodriguez},\ and\
  \citenamefont {Moreno}}]{ferraz2023multistability}%
  \BibitemOpen
  \bibfield  {author} {\bibinfo {author} {\bibfnamefont {G.}~\bibnamefont
  {Ferraz~de Arruda}}, \bibinfo {author} {\bibfnamefont {G.}~\bibnamefont
  {Petri}}, \bibinfo {author} {\bibfnamefont {P.~M.}\ \bibnamefont
  {Rodriguez}},\ and\ \bibinfo {author} {\bibfnamefont {Y.}~\bibnamefont
  {Moreno}},\ }\bibfield  {title} {\bibinfo {title} {Multistability,
  intermittency, and hybrid transitions in social contagion models on
  hypergraphs},\ }\href
  {https://doi.org/https://doi.org/10.1038/s41467-023-37118-3} {\bibfield
  {journal} {\bibinfo  {journal} {Nat. Commun.}\ }\textbf {\bibinfo {volume}
  {14}},\ \bibinfo {pages} {1375} (\bibinfo {year} {2023})}\BibitemShut
  {NoStop}%
\bibitem [{\citenamefont {St-Onge}\ \emph {et~al.}(2021)\citenamefont
  {St-Onge}, \citenamefont {Thibeault}, \citenamefont {Allard}, \citenamefont
  {Dub{\'e}},\ and\ \citenamefont {H{\'e}bert-Dufresne}}]{st2021master}%
  \BibitemOpen
  \bibfield  {author} {\bibinfo {author} {\bibfnamefont {G.}~\bibnamefont
  {St-Onge}}, \bibinfo {author} {\bibfnamefont {V.}~\bibnamefont {Thibeault}},
  \bibinfo {author} {\bibfnamefont {A.}~\bibnamefont {Allard}}, \bibinfo
  {author} {\bibfnamefont {L.~J.}\ \bibnamefont {Dub{\'e}}},\ and\ \bibinfo
  {author} {\bibfnamefont {L.}~\bibnamefont {H{\'e}bert-Dufresne}},\ }\bibfield
   {title} {\bibinfo {title} {Master equation analysis of mesoscopic
  localization in contagion dynamics on higher-order networks},\ }\href
  {https://doi.org/10.1103/PhysRevE.103.032301} {\bibfield  {journal} {\bibinfo
   {journal} {Phys. Rev. E}\ }\textbf {\bibinfo {volume} {103}},\ \bibinfo
  {pages} {032301} (\bibinfo {year} {2021})}\BibitemShut {NoStop}%
\bibitem [{\citenamefont {Iacopini}\ \emph {et~al.}(2022)\citenamefont
  {Iacopini}, \citenamefont {Petri}, \citenamefont {Baronchelli},\ and\
  \citenamefont {Barrat}}]{iacopini2022group}%
  \BibitemOpen
  \bibfield  {author} {\bibinfo {author} {\bibfnamefont {I.}~\bibnamefont
  {Iacopini}}, \bibinfo {author} {\bibfnamefont {G.}~\bibnamefont {Petri}},
  \bibinfo {author} {\bibfnamefont {A.}~\bibnamefont {Baronchelli}},\ and\
  \bibinfo {author} {\bibfnamefont {A.}~\bibnamefont {Barrat}},\ }\bibfield
  {title} {\bibinfo {title} {Group interactions modulate critical mass dynamics
  in social convention},\ }\href
  {https://doi.org/https://doi.org/10.1038/s42005-022-00845-y} {\bibfield
  {journal} {\bibinfo  {journal} {Commun. Phys.}\ }\textbf {\bibinfo {volume}
  {5}},\ \bibinfo {pages} {64} (\bibinfo {year} {2022})}\BibitemShut {NoStop}%
\bibitem [{\citenamefont {Jhun}\ \emph {et~al.}(2019)\citenamefont {Jhun},
  \citenamefont {Jo},\ and\ \citenamefont {Kahng}}]{jhun2019simplicial}%
  \BibitemOpen
  \bibfield  {author} {\bibinfo {author} {\bibfnamefont {B.}~\bibnamefont
  {Jhun}}, \bibinfo {author} {\bibfnamefont {M.}~\bibnamefont {Jo}},\ and\
  \bibinfo {author} {\bibfnamefont {B.}~\bibnamefont {Kahng}},\ }\bibfield
  {title} {\bibinfo {title} {Simplicial {{SIS}} model in scale-free uniform
  hypergraph},\ }\href {https://doi.org/10.1088/1742-5468/ab5367} {\bibfield
  {journal} {\bibinfo  {journal} {J. Stat. Mech.: Theory Exp.}\ }\textbf
  {\bibinfo {volume} {2019}}\bibinfo  {number} { (12)},\ \bibinfo {pages}
  {123207}}\BibitemShut {NoStop}%
\bibitem [{\citenamefont {Malizia}\ \emph {et~al.}(2023)\citenamefont
  {Malizia}, \citenamefont {Gallo}, \citenamefont {Frasca}, \citenamefont
  {Latora},\ and\ \citenamefont {Russo}}]{malizia2023pair}%
  \BibitemOpen
\bibfield  {number} {  }\bibfield  {author} {\bibinfo {author} {\bibfnamefont
  {F.}~\bibnamefont {Malizia}}, \bibinfo {author} {\bibfnamefont
  {L.}~\bibnamefont {Gallo}}, \bibinfo {author} {\bibfnamefont
  {M.}~\bibnamefont {Frasca}}, \bibinfo {author} {\bibfnamefont
  {V.}~\bibnamefont {Latora}},\ and\ \bibinfo {author} {\bibfnamefont
  {G.}~\bibnamefont {Russo}},\ }\bibfield  {title} {\bibinfo {title} {A
  pair-based approximation for simplicial contagion},\ }\bibfield  {journal}
  {\bibinfo  {journal} {arXiv preprint arXiv:2307.10151}\ }\href
  {https://doi.org/https://doi.org/10.48550/arXiv.2307.10151}
  {https://doi.org/10.48550/arXiv.2307.10151} (\bibinfo {year}
  {2023})\BibitemShut {NoStop}%
\bibitem [{\citenamefont {Burgio}\ \emph
  {et~al.}(2023{\natexlab{a}})\citenamefont {Burgio}, \citenamefont
  {G{\'o}mez},\ and\ \citenamefont {Arenas}}]{burgio2023triadic}%
  \BibitemOpen
  \bibfield  {author} {\bibinfo {author} {\bibfnamefont {G.}~\bibnamefont
  {Burgio}}, \bibinfo {author} {\bibfnamefont {S.}~\bibnamefont {G{\'o}mez}},\
  and\ \bibinfo {author} {\bibfnamefont {A.}~\bibnamefont {Arenas}},\
  }\bibfield  {title} {\bibinfo {title} {Triadic approximation for contagions
  on higher-order networks},\ }\bibfield  {journal} {\bibinfo  {journal} {arXiv
  preprint arXiv:2306.11441}\ }\href
  {https://doi.org/https://doi.org/10.48550/arXiv.2306.11441}
  {https://doi.org/10.48550/arXiv.2306.11441} (\bibinfo {year}
  {2023}{\natexlab{a}})\BibitemShut {NoStop}%
\bibitem [{\citenamefont {Burgio}\ \emph {et~al.}(2021)\citenamefont {Burgio},
  \citenamefont {Arenas}, \citenamefont {G{\'o}mez},\ and\ \citenamefont
  {Matamalas}}]{burgio2021network}%
  \BibitemOpen
  \bibfield  {author} {\bibinfo {author} {\bibfnamefont {G.}~\bibnamefont
  {Burgio}}, \bibinfo {author} {\bibfnamefont {A.}~\bibnamefont {Arenas}},
  \bibinfo {author} {\bibfnamefont {S.}~\bibnamefont {G{\'o}mez}},\ and\
  \bibinfo {author} {\bibfnamefont {J.~T.}\ \bibnamefont {Matamalas}},\
  }\bibfield  {title} {\bibinfo {title} {Network clique cover approximation to
  analyze complex contagions through group interactions},\ }\href
  {https://doi.org/https://doi.org/10.1038/s42005-021-00618-z} {\bibfield
  {journal} {\bibinfo  {journal} {Commun. Phys.}\ }\textbf {\bibinfo {volume}
  {4}},\ \bibinfo {pages} {1} (\bibinfo {year} {2021})}\BibitemShut {NoStop}%
\bibitem [{\citenamefont {Burgio}\ \emph
  {et~al.}(2023{\natexlab{b}})\citenamefont {Burgio}, \citenamefont {St-Onge},\
  and\ \citenamefont {H{\'e}bert-Dufresne}}]{burgio2023adaptive}%
  \BibitemOpen
  \bibfield  {author} {\bibinfo {author} {\bibfnamefont {G.}~\bibnamefont
  {Burgio}}, \bibinfo {author} {\bibfnamefont {G.}~\bibnamefont {St-Onge}},\
  and\ \bibinfo {author} {\bibfnamefont {L.}~\bibnamefont
  {H{\'e}bert-Dufresne}},\ }\bibfield  {title} {\bibinfo {title} {Adaptive
  hypergraphs and the characteristic scale of higher-order contagions using
  generalized approximate master equations},\ }\bibfield  {journal} {\bibinfo
  {journal} {arXiv preprint arXiv:2307.11268}\ }\href
  {https://doi.org/https://doi.org/10.48550/arXiv.2307.11268}
  {https://doi.org/10.48550/arXiv.2307.11268} (\bibinfo {year}
  {2023}{\natexlab{b}})\BibitemShut {NoStop}%
\bibitem [{\citenamefont {Salnikov}\ \emph {et~al.}(2018)\citenamefont
  {Salnikov}, \citenamefont {Cassese},\ and\ \citenamefont
  {Lambiotte}}]{salnikov2018simplicial}%
  \BibitemOpen
  \bibfield  {author} {\bibinfo {author} {\bibfnamefont {V.}~\bibnamefont
  {Salnikov}}, \bibinfo {author} {\bibfnamefont {D.}~\bibnamefont {Cassese}},\
  and\ \bibinfo {author} {\bibfnamefont {R.}~\bibnamefont {Lambiotte}},\
  }\bibfield  {title} {\bibinfo {title} {Simplicial complexes and complex
  systems},\ }\href {https://doi.org/10.1088/1361-6404/aae790} {\bibfield
  {journal} {\bibinfo  {journal} {Eur. J. Phys.}\ }\textbf {\bibinfo {volume}
  {40}},\ \bibinfo {pages} {014001} (\bibinfo {year} {2018})}\BibitemShut
  {NoStop}%
\bibitem [{\citenamefont {Aleksandrov}(1998)}]{aleksandrov1998combinatorial}%
  \BibitemOpen
  \bibfield  {author} {\bibinfo {author} {\bibfnamefont {P.}~\bibnamefont
  {Aleksandrov}},\ }\href {https://books.google.co.uk/books?id=pZ3HcV9cPOAC}
  {\emph {\bibinfo {title} {Combinatorial Topology}}},\ \bibinfo {series}
  {Dover Books on Mathematics}\ No.\ \bibinfo {number} {v. 1-2}\ (\bibinfo
  {publisher} {Dover Publications},\ \bibinfo {year} {1998})\BibitemShut
  {NoStop}%
\bibitem [{\citenamefont {Hatcher}(2002)}]{hatcher2002algebraic}%
  \BibitemOpen
  \bibfield  {author} {\bibinfo {author} {\bibfnamefont {A.}~\bibnamefont
  {Hatcher}},\ }\href {https://books.google.co.uk/books?id=BjKs86kosqgC} {\emph
  {\bibinfo {title} {Algebraic Topology}}},\ Algebraic Topology\ (\bibinfo
  {publisher} {Cambridge University Press},\ \bibinfo {year}
  {2002})\BibitemShut {NoStop}%
\bibitem [{\citenamefont {Gladwell}(2001)}]{gladwell2001tipping}%
  \BibitemOpen
  \bibfield  {author} {\bibinfo {author} {\bibfnamefont {M.}~\bibnamefont
  {Gladwell}},\ }\href {https://books.google.co.uk/books?id=GqepQgAACAAJ}
  {\emph {\bibinfo {title} {The Tipping Point: How Little Things Can Make a Big
  Difference}}},\ A Back bay Book\ (\bibinfo  {publisher} {Back Bay Books},\
  \bibinfo {year} {2001})\BibitemShut {NoStop}%
\bibitem [{\citenamefont {Simon}\ \emph {et~al.}(1999)\citenamefont {Simon},
  \citenamefont {Farkas},\ and\ \citenamefont
  {Wittmann}}]{simon1999constructing}%
  \BibitemOpen
  \bibfield  {author} {\bibinfo {author} {\bibfnamefont {P.~L.}\ \bibnamefont
  {Simon}}, \bibinfo {author} {\bibfnamefont {H.}~\bibnamefont {Farkas}},\ and\
  \bibinfo {author} {\bibfnamefont {M.}~\bibnamefont {Wittmann}},\ }\bibfield
  {title} {\bibinfo {title} {Constructing global bifurcation diagrams by the
  parametric representation method},\ }\href
  {https://doi.org/https://doi.org/10.1016/S0377-0427(99)00108-9} {\bibfield
  {journal} {\bibinfo  {journal} {J. Comput. Appl. Math.}\ }\textbf {\bibinfo
  {volume} {108}},\ \bibinfo {pages} {157} (\bibinfo {year}
  {1999})}\BibitemShut {NoStop}%
\bibitem [{\citenamefont {Perra}\ \emph {et~al.}(2011)\citenamefont {Perra},
  \citenamefont {Balcan}, \citenamefont {Gon{\c{c}}alves},\ and\ \citenamefont
  {Vespignani}}]{perra2011towards}%
  \BibitemOpen
  \bibfield  {author} {\bibinfo {author} {\bibfnamefont {N.}~\bibnamefont
  {Perra}}, \bibinfo {author} {\bibfnamefont {D.}~\bibnamefont {Balcan}},
  \bibinfo {author} {\bibfnamefont {B.}~\bibnamefont {Gon{\c{c}}alves}},\ and\
  \bibinfo {author} {\bibfnamefont {A.}~\bibnamefont {Vespignani}},\ }\bibfield
   {title} {\bibinfo {title} {Towards a characterization of behavior-disease
  models},\ }\href
  {https://doi.org/https://doi.org/10.1371/journal.pone.0023084} {\bibfield
  {journal} {\bibinfo  {journal} {PLoS One}\ }\textbf {\bibinfo {volume} {6}},\
  \bibinfo {pages} {e23084} (\bibinfo {year} {2011})}\BibitemShut {NoStop}%
\bibitem [{\citenamefont {Scarpino}\ \emph {et~al.}(2016)\citenamefont
  {Scarpino}, \citenamefont {Allard},\ and\ \citenamefont
  {H{\'e}bert-Dufresne}}]{scarpino2016effect}%
  \BibitemOpen
  \bibfield  {author} {\bibinfo {author} {\bibfnamefont {S.~V.}\ \bibnamefont
  {Scarpino}}, \bibinfo {author} {\bibfnamefont {A.}~\bibnamefont {Allard}},\
  and\ \bibinfo {author} {\bibfnamefont {L.}~\bibnamefont
  {H{\'e}bert-Dufresne}},\ }\bibfield  {title} {\bibinfo {title} {The effect of
  a prudent adaptive behaviour on disease transmission},\ }\href
  {https://doi.org/https://doi.org/10.1038/nphys3832} {\bibfield  {journal}
  {\bibinfo  {journal} {Nat. Phys.}\ }\textbf {\bibinfo {volume} {12}},\
  \bibinfo {pages} {1042} (\bibinfo {year} {2016})}\BibitemShut {NoStop}%
\bibitem [{\citenamefont {Perra}(2021)}]{perra2021non}%
  \BibitemOpen
  \bibfield  {author} {\bibinfo {author} {\bibfnamefont {N.}~\bibnamefont
  {Perra}},\ }\bibfield  {title} {\bibinfo {title} {Non-pharmaceutical
  interventions during the covid-19 pandemic: A review},\ }\href
  {https://doi.org/https://doi.org/10.1016/j.physrep.2021.02.001} {\bibfield
  {journal} {\bibinfo  {journal} {Phys. Repo.}\ }\textbf {\bibinfo {volume}
  {913}},\ \bibinfo {pages} {1} (\bibinfo {year} {2021})}\BibitemShut {NoStop}%
\bibitem [{\citenamefont {Lucas}\ \emph {et~al.}(2023)\citenamefont {Lucas},
  \citenamefont {Iacopini}, \citenamefont {Robiglio}, \citenamefont {Barrat},\
  and\ \citenamefont {Petri}}]{lucas2023simplicially}%
  \BibitemOpen
  \bibfield  {author} {\bibinfo {author} {\bibfnamefont {M.}~\bibnamefont
  {Lucas}}, \bibinfo {author} {\bibfnamefont {I.}~\bibnamefont {Iacopini}},
  \bibinfo {author} {\bibfnamefont {T.}~\bibnamefont {Robiglio}}, \bibinfo
  {author} {\bibfnamefont {A.}~\bibnamefont {Barrat}},\ and\ \bibinfo {author}
  {\bibfnamefont {G.}~\bibnamefont {Petri}},\ }\bibfield  {title} {\bibinfo
  {title} {Simplicially driven simple contagion},\ }\href
  {https://doi.org/https://doi.org/10.1103/PhysRevResearch.5.013201} {\bibfield
   {journal} {\bibinfo  {journal} {Phys. Rev. Res.}\ }\textbf {\bibinfo
  {volume} {5}},\ \bibinfo {pages} {013201} (\bibinfo {year}
  {2023})}\BibitemShut {NoStop}%
\bibitem [{\citenamefont {Holme}\ and\ \citenamefont
  {Saram{\"a}ki}(2012)}]{holme2012temporal}%
  \BibitemOpen
  \bibfield  {author} {\bibinfo {author} {\bibfnamefont {P.}~\bibnamefont
  {Holme}}\ and\ \bibinfo {author} {\bibfnamefont {J.}~\bibnamefont
  {Saram{\"a}ki}},\ }\bibfield  {title} {\bibinfo {title} {Temporal networks},\
  }\href {https://doi.org/https://doi.org/10.1016/j.physrep.2012.03.001}
  {\bibfield  {journal} {\bibinfo  {journal} {Phys. Rep.}\ }\textbf {\bibinfo
  {volume} {519}},\ \bibinfo {pages} {97} (\bibinfo {year} {2012})}\BibitemShut
  {NoStop}%
\bibitem [{\citenamefont {Cencetti}\ \emph {et~al.}(2021)\citenamefont
  {Cencetti}, \citenamefont {Battiston}, \citenamefont {Lepri},\ and\
  \citenamefont {Karsai}}]{cencetti2021temporal}%
  \BibitemOpen
  \bibfield  {author} {\bibinfo {author} {\bibfnamefont {G.}~\bibnamefont
  {Cencetti}}, \bibinfo {author} {\bibfnamefont {F.}~\bibnamefont {Battiston}},
  \bibinfo {author} {\bibfnamefont {B.}~\bibnamefont {Lepri}},\ and\ \bibinfo
  {author} {\bibfnamefont {M.}~\bibnamefont {Karsai}},\ }\bibfield  {title}
  {\bibinfo {title} {Temporal properties of higher-order interactions in social
  networks},\ }\href
  {https://doi.org/https://doi.org/10.1038/s41598-021-86469-8} {\bibfield
  {journal} {\bibinfo  {journal} {Sci. Rep.}\ }\textbf {\bibinfo {volume}
  {11}},\ \bibinfo {pages} {7028} (\bibinfo {year} {2021})}\BibitemShut
  {NoStop}%
\bibitem [{\citenamefont {Gallo}\ \emph {et~al.}(2023)\citenamefont {Gallo},
  \citenamefont {Lacasa}, \citenamefont {Latora},\ and\ \citenamefont
  {Battiston}}]{gallo2023higher}%
  \BibitemOpen
  \bibfield  {author} {\bibinfo {author} {\bibfnamefont {L.}~\bibnamefont
  {Gallo}}, \bibinfo {author} {\bibfnamefont {L.}~\bibnamefont {Lacasa}},
  \bibinfo {author} {\bibfnamefont {V.}~\bibnamefont {Latora}},\ and\ \bibinfo
  {author} {\bibfnamefont {F.}~\bibnamefont {Battiston}},\ }\bibfield  {title}
  {\bibinfo {title} {Higher-order correlations reveal complex memory in
  temporal hypergraphs},\ }\href@noop {} {\bibfield  {journal} {\bibinfo
  {journal} {arXiv preprint arXiv:2303.09316}\ } (\bibinfo {year}
  {2023})}\BibitemShut {NoStop}%
\bibitem [{\citenamefont {Iacopini}\ \emph
  {et~al.}(2023{\natexlab{b}})\citenamefont {Iacopini}, \citenamefont
  {Karsai},\ and\ \citenamefont {Barrat}}]{iacopini2023temporal}%
  \BibitemOpen
  \bibfield  {author} {\bibinfo {author} {\bibfnamefont {I.}~\bibnamefont
  {Iacopini}}, \bibinfo {author} {\bibfnamefont {M.}~\bibnamefont {Karsai}},\
  and\ \bibinfo {author} {\bibfnamefont {A.}~\bibnamefont {Barrat}},\
  }\bibfield  {title} {\bibinfo {title} {The temporal dynamics of group
  interactions in higher-order social networks},\ }\bibfield  {journal}
  {\bibinfo  {journal} {arXiv preprint arXiv:2306.09967}\ }\href
  {https://doi.org/https://doi.org/10.48550/arXiv.2306.09967}
  {https://doi.org/10.48550/arXiv.2306.09967} (\bibinfo {year}
  {2023}{\natexlab{b}})\BibitemShut {NoStop}%
\bibitem [{\citenamefont {Chowdhary}\ \emph {et~al.}(2021)\citenamefont
  {Chowdhary}, \citenamefont {Kumar}, \citenamefont {Cencetti}, \citenamefont
  {Iacopini},\ and\ \citenamefont {Battiston}}]{chowdhary2021simplicial}%
  \BibitemOpen
  \bibfield  {author} {\bibinfo {author} {\bibfnamefont {S.}~\bibnamefont
  {Chowdhary}}, \bibinfo {author} {\bibfnamefont {A.}~\bibnamefont {Kumar}},
  \bibinfo {author} {\bibfnamefont {G.}~\bibnamefont {Cencetti}}, \bibinfo
  {author} {\bibfnamefont {I.}~\bibnamefont {Iacopini}},\ and\ \bibinfo
  {author} {\bibfnamefont {F.}~\bibnamefont {Battiston}},\ }\bibfield  {title}
  {\bibinfo {title} {Simplicial contagion in temporal higher-order networks},\
  }\href {https://doi.org/https://doi.org/10.1088/2632-072X/ac12bd} {\bibfield
  {journal} {\bibinfo  {journal} {J. phys. Complex.}\ }\textbf {\bibinfo
  {volume} {2}},\ \bibinfo {pages} {035019} (\bibinfo {year}
  {2021})}\BibitemShut {NoStop}%
\bibitem [{\citenamefont {Neuh{\"a}user}\ \emph {et~al.}(2021)\citenamefont
  {Neuh{\"a}user}, \citenamefont {Lambiotte},\ and\ \citenamefont
  {Schaub}}]{neuhauser2021consensus}%
  \BibitemOpen
  \bibfield  {author} {\bibinfo {author} {\bibfnamefont {L.}~\bibnamefont
  {Neuh{\"a}user}}, \bibinfo {author} {\bibfnamefont {R.}~\bibnamefont
  {Lambiotte}},\ and\ \bibinfo {author} {\bibfnamefont {M.~T.}\ \bibnamefont
  {Schaub}},\ }\bibfield  {title} {\bibinfo {title} {Consensus dynamics on
  temporal hypergraphs},\ }\href {https://doi.org/10.1103/PhysRevE.104.064305}
  {\bibfield  {journal} {\bibinfo  {journal} {Phys. Rev. E}\ }\textbf {\bibinfo
  {volume} {104}},\ \bibinfo {pages} {064305} (\bibinfo {year}
  {2021})}\BibitemShut {NoStop}%
\bibitem [{\citenamefont {Rudin}(1976)}]{rudin1953principles}%
  \BibitemOpen
  \bibfield  {author} {\bibinfo {author} {\bibfnamefont {W.}~\bibnamefont
  {Rudin}},\ }\href@noop {} {\emph {\bibinfo {title} {Principles of
  mathematical analysis, 3rd edition}}}\ (\bibinfo {year} {1976})\BibitemShut
  {NoStop}%
\end{thebibliography}%

\subsection*{Authors' contributions}
I.Z.K.: conceptualisation, formal analysis, investigation, methodology, software, validation, visualisation, writing---original draft, writing---review and editing; 
I.I.: formal analysis, investigation, visualisation, writing---original draft, writing---review and editing;
P.L.S.: formal analysis, methodology, writing---review and editing;
N.G.: formal analysis, methodology, writing---review and editing.
All authors gave final approval for publication and agreed to be held accountable for the work performed therein.

\subsection*{Competing interests}
We declare that we have no competing interests.

\subsection*{Funding}
N. Georgiou was partially supported by the Dr. Perry James Browne Research Centre, at the University of Sussex. P.L. Simon acknowledges support from the Hungarian Scientific Research Fund, OTKA (grant no. 135241) and from the Ministry of Innovation and Technology NRDI Office within the framework of the Artificial Intelligence National Laboratory Programme.

\end{document}